\documentclass[a4paper]{article}

\usepackage[utf8]{inputenc}

\usepackage[dvipdfmx]{graphicx}
\usepackage{bmpsize}
\usepackage{multicol}
\usepackage{amsthm}
\usepackage{bm}
\usepackage{amssymb}
\usepackage{amsmath}
\usepackage{color}
\usepackage{amsthm}
\usepackage{subcaption} 
\usepackage{graphicx}
\usepackage{dcolumn}
\usepackage{bm}

\usepackage[utf8]{inputenc}
\usepackage[T1]{fontenc}
\usepackage{etoolbox}
\usepackage{ascmac,wrapfig,makeidx}
\usepackage{physics}
\usepackage[stable]{footmisc}
\usepackage{mathrsfs}
\usepackage{color}
\usepackage{comment}
\usepackage{amsmath}

\theoremstyle{plain}
\newtheorem{theorem}{Theorem}
\newtheorem*{theorem*}{Theorem}
\newtheorem{lemma}{Lemma}[section]
\newtheorem{corollary}{Corollary}

\theoremstyle{definition}

\newcommand{\BA}{\begin{eqnarray}}
\newcommand{\EA}{\end{eqnarray}}

\definecolor{dgreen}{rgb}{0.0, 0.5, 0.0}

\begin{document}

\fontsize{14pt}{16.5pt}\selectfont

\begin{center}
\bf{Rigged Hilbert Space formulation for quasi-Hermitian \\composite systems
}
\end{center}
\fontsize{12pt}{11pt}\selectfont
\begin{center}
Shousuke Ohmori\\ 
\end{center}

\bigskip

\noindent
\it{National Institute of Technology, Gunma College, 
    580 Toribamachi, Maebashi-shi, Gunma 371-8530, Japan,
}\\
\it{
Waseda Research Institute for Science and Engineering, Waseda University, 
    Shinjuku, Tokyo 169-8555, Japan.
}

\bigskip

\noindent
*corresponding author: 42261timemachine@ruri.waseda.jp\\
~~\\
\rm
\fontsize{11pt}{14pt}\selectfont\noindent

\baselineskip 20pt

{\bf Abstract}\\
%
The discussion in this study delves into Dirac's bra-ket formalism for a quasi-Hermitian quantum composite system based on the rigged Hilbert space (RHS).
We establish an RHS with a positive-definite metric suitable for a quasi-Hermitian composite system. 
The obtained RHS is utilized to construct the bra and ket vectors for the non-Hermitian composite system and produce the spectral decomposition of the quasi-Hermitian operator.
We show that the symmetric relations regarding quasi-Hermitian operators can be extended to dual spaces, and all descriptions obtained using the bra-ket formalism are completely developed in the dual spaces.
Our methodology is applied to a non-Hermitian harmonic oscillator composed of conformal
multi-dimensional many-body systems.


\section{Introduction}
\label{sec:1}

The precise handling of Dirac's bra-ket notation mathematically is a critical issue in constructing a mathematical description in quantum mechanics because it is considered insufficient within von Neumann's Hilbert space theory.
The rigged Hilbert space (RHS) approach has been developed to resolve this issue\cite{Robert1966a,Robert1966b,Antoine1969a,Antoine1969b,Melsheimer1974a,Melsheimer1974b,Bohm1978,Bohm1981,Prigogine1996,Bohm1998,Antoiou1998,Antoiou2003,Gadella2003,Madrid2004,Madrid2005,Antoine2009,Antoine2021,Ohmori2024}.
Actually, by applying the nuclear spectral theorem\cite{Maurin1968} to an Hermitian operator (observable) in RHS, 
eigenequations for the bra and ket vectors are formulated, individually.
This theorem also provides the spectral expansions that are identified with the spectral decomposition for the discrete and continuous spectra specified by
the Dirac's $\delta$-function (distributions) found in the literature.
In this sense, the Dirac's bra-ket formalism is completely obtained, 
and nowadays, studies comprising the use of the RHS are conducted not only with a focus on  quantum systems but also on a wide array of physical problems, such as the construction of generalized spectral decompositions for operators used in dynamical systems\cite{Antoniou1993,Suchanecki1996,Chiba2015,Chiba2023}.

The RHS is also effective in the mathematical treatment of modern quantum physics, such as non-equilibrium open systems and non-Hermitian quantum systems\cite{Chruscinski2003,Chruscinski2004,Lars2019,Fernandez2022,Ohmori2022}.
In these systems, we encounter the physical problems that
can not be addressed within the framework of the Hilbert space alone, such as the $L^2$-space. 
For instance, in the problem of a quantum damped system,
the Hamiltonian exhibits only real spectra in the $L^2$-space.
However, when the RHS is selected as the fundamental space, the Hamiltonian exhibits complex eigenvalues, which are interpreted as the resonant states\cite{Chruscinski2003,Chruscinski2004}. 
This example suggests that the RHS is indispensable for addressing complex eigenvalues beyond the $L^2$-space theory.
Another illustrative example involving the RHS is non-Hermitian quantum systems whose non-Hermitian operator ${A}$ is assumed to possess the characteristic symmetric relation,
\begin{equation}
{A}^\dagger =\eta {A} \eta^{-1},
    \label{eqn:O1}
\end{equation}
where $\eta$ is the intertwining operator\cite{Lars2019,Fernandez2022,Ohmori2022}.
This type of operator is found in several non-Hermitian quantum systems, such as 
$\mathcal{PT}$-symmetrical systems\cite{Bender1998,Bender2007}.
The RHS is considered advantageous not only for providing an accurate description of bra-ket formalism, but also for analyzing the eigenvalues of the non-Hermitian operators of the system.

Recently, we presented a general framework for  describing the bra-ket formalism of a non-Hermitian system with the positive-definite metric drawn from the RHS approach\cite{Ohmori2022}. 
A non-Hermitian system with the positive-definite metric is the system characterized by a non-Hermitian operator that satisfies symmetry (\ref{eqn:O1}) where the intertwining operator $\eta$ is assumed to be a positive operator.
Such an operator is referred to as an $\eta$-quasi Hermitian\cite{Dieudonne1961,Mos2010,Antoine2013}.
Utilizing the RHS treatment,
\begin{equation}
    \Phi \subset \mathcal{H} \subset \Phi^\prime,\Phi^\times,
    \label{eqn:RHS}
\end{equation}
proposed by Madrid\cite{Madrid2005}, 
we established another RHS suitable for an $\eta$-quasi Hermitian called the $\eta$-RHS\cite{Ohmori2022},
\begin{equation}
    \Phi \subset \widetilde{\mathcal{H}_\eta} \subset \Phi^\prime,\Phi^\times.
    \label{eqn:RHS_eta}
\end{equation}
In (\ref{eqn:RHS}) and (\ref{eqn:RHS_eta}), 
$\mathcal{H}=(\mathcal{H}, \langle \cdot, \cdot \rangle_\mathcal{H})$ is a complex Hilbert space, and $\Phi=(\Phi, \tau_\Phi)$ is a nuclear space that is a dense linear subspace of $\mathcal{H}$ where the inner product $\langle \cdot, \cdot \rangle_\Phi$ on $\Phi$, $\langle \phi, \psi \rangle_\Phi \equiv \langle \phi, \psi \rangle_\mathcal{H}$ for $\phi, \psi \in \Phi$,  becomes separately continuous on $(\Phi, \tau_\Phi)$.
$\Phi^\prime$ is a family of continuous linear functionals on $(\Phi,\tau_\Phi)$, and $\Phi^\times$ is a family of continuous {\it anti-linear} functionals on $(\Phi,\tau_\Phi)$
(the function $f\in \Phi^{\times}$ is anti-linear if it satisfies $f(a\varphi+b\phi)=a^*f(\varphi)+b^*f(\phi)$, where $a$ and $b$ are complex numbers with complex conjugates $a^*$ and $b^*$ and $\varphi, \phi \in \Phi$).
$\widetilde{\mathcal{H}_\eta}=(\widetilde{\mathcal{H}_\eta},\widetilde{\langle \cdot, \cdot \rangle_\eta})$ is the completion of the pre-Hilbert space $\mathcal{H}_\eta=(\mathcal{H},\langle \cdot, \cdot \rangle_\eta)$ equipping the metric induced from the inner product $\langle \phi, \psi \rangle_\eta
    = \langle \phi, \eta \psi \rangle_\mathcal{H}$ for $\phi,\psi \in \mathcal{H}$ with respect to $\eta$. 
In Madrids' RHS treatment (\ref{eqn:RHS}), the bra and ket vectors are built as elements of $\Phi^\prime$ and $\Phi^\times$.
For $\varphi \in \Phi$, 
a map $\ket{\varphi}_\mathcal{H} : \Phi \rightarrow \mathbb{C}^1$, $\ket{\varphi}_\mathcal{H}(\phi) \equiv  \langle \phi,\, \varphi\rangle_\mathcal{H}$ for $\phi \in \Phi$ is defined; this map is called a ket of $\varphi$.
The bra vector of $\varphi$ is defined as the complex conjugate of $\ket{\varphi}_\mathcal{H}$, 
namely, the map $\bra{\varphi}_\mathcal{H} : \Phi \rightarrow \mathbb{C}^1 $ where $\bra{\varphi}_{\mathcal{H}}(\phi)= \ket{\varphi}_{\mathcal{H}}^*(\phi)=(\ket{\varphi}_{\mathcal{H}}(\phi))^*=\langle \varphi,\, \phi\rangle_\mathcal{H}$.
%
%
It is apparent that $\bra{\varphi}_{\mathcal{H}}$ and $\ket{\varphi}_{\mathcal{H}}$ belong to $\Phi^{\prime}$ and $\Phi^{\times}$ of $\Phi$, respectively.
In the $\eta$-RHS (\ref{eqn:RHS_eta}), an $\eta$-quasi Hermitian becomes Hermitian, 
and hence the nuclear spectral theorem can be executed for obtaining the spectral expansions for the bra and ket vectors. 
Calculations are then performed in the dual and anti-dual spaces $\Phi^{\prime}$ and $\Phi^{\times}$ using these expansions.
Furthermore, the spectral expansions for the original system (\ref{eqn:RHS}) are presented by using extension of $\eta$ to the dual spaces. 
We consider that the description of the bra-ket expanded under the dual spaces is a critical point when treating the RHS.

Building on our previous study, this study is aimed at developing the $\eta$-RHS to accommodate non-Hermitian composite systems.
This development is necessary for mathematically addressing non-Hermitian quantum systems described using quantum statistical mechanics and quantum field theory.
The bra-ket formalism is presented by focusing on non-Hermitian composite systems with the positive-definite metric.
Moreover, we describe the formulation of the quasi-Hermitian operator in a composite system based on dual spaces, as all the bra-ket calculations, such as spectral expansions, are executed in the dual spaces. 
This formulation has the advantage to settle the issue of characterizing the adjoint operator of a composite operator in a non-Hermitian system.

The remainder of this paper is organized as follows.
In Section~\ref{sec:2}, we establish an $\eta_1\otimes\eta_2$-RHS, which defines a bra and ket vector describing the composite system using the tensor product of a single $\eta$-RHS.
Moreover, we proceed to demonstrate the extension of the positive metric $\eta_1\otimes\eta_2$, which converts the bra and ket vectors in the non-Hermitian system into those of an Hermitian system and vice versa. 
Section~\ref{sec:3} describes the mathematical formulation of the spectral expansion of the bra and ket vectors with respect to the quasi-Hermitian operator for a composite system. 
After proving some properties of the quasi-Hermitian operator,
we show that the bra and ket expansions for the composite system can be performed in dual and anti-dual spaces using the generalized eigenvectors obtained as the tensor products of the bras and kets of single systems, respectively.
The procedure for transferring a quasi-Hermitian with its adjoint to the dual spaces is shown in Section \ref{sec:4}.
The results obtained from the previous sections and formulation of the bra-ket vectors and their expansion are then consistently described in the dual spaces.
The methodology of the dual spaces is applied to a non-Hermitian harmonic oscillator composed of conformal
multi-dimensional many-body systems in Section~\ref{sec:5}.
Finally, the conclusions are presented in Section~\ref{sec:6}.

\section{Construction of $\eta_1\otimes \eta_2$-RHS}
\label{sec:2}

\subsection{$\eta_1\otimes \eta_2$-RHS}
\label{sec:2.1}

In the context of the RHS approach,
a composite system is described by the tensor products of the RHSs\cite{Robert1966b,Melsheimer1974a,Ohmori2024}.
Let $\Phi_i \subset \mathcal{H}_i \subset \Phi _i^\prime, \Phi _i^\times $ be an RHS, (\ref{eqn:RHS}), where we set $i=1,2$ for simplicity.
The tensor product of these spaces that becomes an RHS is given by
\begin{eqnarray}
    \Phi_1 \hat{\otimes} {\Phi}_2 \subset \mathcal{H}_1 \bar {\otimes} \mathcal{H}_2 \subset 
    (\Phi_1 \hat{\otimes} {\Phi}_2)^{\prime},~(\Phi_1 \hat{\otimes} {\Phi}_2)^{\times}.
    \label{eqn:2-00}
\end{eqnarray}
Here, 
$\mathcal{H}_1 \bar {\otimes} \mathcal{H}_2=(\mathcal{H}_1 \bar{\otimes} \mathcal{H}_2,\langle \cdot,\, \cdot\rangle_{\mathcal{H}_1 \bar{\otimes} \mathcal{H}_2})$ is the completion of the algebraic tensor product $\mathcal{H}_1 \otimes \mathcal{H}_2=\Big{\{}\displaystyle\sum_{j=1}^m\varphi_{1j}\otimes\varphi_{2j}\mid \varphi_{1j}\in \mathcal{H}_1,\varphi_{2j}\in \mathcal{H}_2, j=1\sim m, m\in \mathbb{N} \Big{\}}$ with respect to the topology induced by $\langle \cdot,\, \cdot\rangle_{\mathcal{H}_1 \otimes \mathcal{H}_2}$; 
this inner product satisfies $\langle \varphi_1 \otimes \varphi_2,\, \phi_1\otimes\phi_2\rangle_{\mathcal{H}_1 \otimes \mathcal{H}_2}=\langle \varphi_1, \, \phi_1 \rangle_{\mathcal{H}_1}\langle \varphi_2, \, \phi_2 \rangle_{\mathcal{H}_2}$.
$\Phi_1 \hat{\otimes} {\Phi}_2=(\Phi_1 \hat{\otimes} {\Phi}_2,\widehat{\tau_p})$ is 
the completion of $(\Phi_1 \otimes {\Phi}_2,\tau_p)$ where  $\Phi_1 \otimes {\Phi}_2=\Big{\{}\displaystyle\sum_{j=1}^m\varphi_{1j}\otimes\varphi_{2j}\mid \varphi_{1j}\in {\Phi}_1,\varphi_{2j}\in {\Phi}_2, j=1\sim m, m\in \mathbb{N} \Big{\}}$ is the algebraic tensor product of the nuclear spaces $(\Phi_1,\tau_{\Phi_1})$ and $(\Phi_2,\tau_{\Phi_2})$ equipping the locally convex topology $\tau_p$ with the local base $\mathcal{B}_p=\{\Gamma(V_1\otimes V_2) \mid V_i\in \mathcal{B}_i,i=1,2\}$ where each $\mathcal{B}_i$ is a local base of $\tau_{\Phi_i}$ and $\Gamma(X)$ stands for the convex circled hull of a set $X$\cite{Schaefer1966}.
Not that $\Phi_1 \hat{\otimes} {\Phi}_2=(\Phi_1 \hat{\otimes} {\Phi}_2,\widehat{\tau_p})$ is a nuclear space.
$(\Phi_1 \hat{\otimes} {\Phi}_2)^{\prime}$ and $(\Phi_1 \hat{\otimes} {\Phi}_2)^{\times}$ are the dual and anti-dual spaces, respectively.
The bra and ket vectors corresponding to $\varphi\in \Phi_1 \hat{\otimes} {\Phi}_2$ are defined as
\begin{eqnarray}
    \bra{\varphi}_{\mathcal{H}_1 \bar{\otimes} \mathcal{H}_2} : \Phi_1 \hat{\otimes} {\Phi}_2 \to \mathbb{C},~\bra{\varphi}_{\mathcal{H}_1 \bar{\otimes} \mathcal{H}_2}(\phi)=\langle \varphi,\, \phi\rangle_{\mathcal{H}_1 \bar {\otimes} \mathcal{H}_2}, \nonumber\\
    \ket{\varphi}_{\mathcal{H}_1 \bar{\otimes} \mathcal{H}_2} : \Phi_1 \hat{\otimes} {\Phi}_2 \to \mathbb{C},~\ket{\varphi}_{\mathcal{H}_1 \bar{\otimes} \mathcal{H}_2}(\phi)=\langle \phi,\, \varphi\rangle_{\mathcal{H}_1 \bar {\otimes} \mathcal{H}_2}.    
    \label{eqn:bra-ket_definition}
\end{eqnarray}
It is apparent that the relations $\ket{\varphi}_{\mathcal{H}_1 \bar{\otimes} \mathcal{H}_2}=\bra{\varphi}_{\mathcal{H}_1 \bar{\otimes} \mathcal{H}_2}^*$, $\bra{\varphi}_{\mathcal{H}_1 \bar{\otimes} \mathcal{H}_2} \in (\Phi_1 \hat{\otimes} {\Phi}_2)^{\prime}$, and $\ket{\varphi}_{\mathcal{H}_1 \bar{\otimes} \mathcal{H}_2} \in (\Phi_1 \hat{\otimes} {\Phi}_2)^{\times}$, are satisfied.
Besides, for any $\varphi_1 \in \Phi_1$ and $\varphi_2 \in \Phi_2$, the relations 
\begin{eqnarray}
    \bra{\varphi_1 \otimes \varphi_2}_{\mathcal{H}_1 \bar{\otimes} \mathcal{H}_2}=\bra{\varphi_1}_{\mathcal{H}_1} \otimes \bra{\varphi_2}_{\mathcal{H}_2},~~
    \ket{\varphi_1 \otimes \varphi_2}_{\mathcal{H}_1 \bar{\otimes} \mathcal{H}_2}=\ket{\varphi_1}_{\mathcal{H}_1} \otimes \ket{\varphi_2}_{\mathcal{H}_2},
\label{eqn:bra_ket_relation}
\end{eqnarray}
 are obtained\cite{Ohmori2024}.
They show that the bra (ket) vector for the type $\varphi=\varphi_1\otimes \varphi_2$ can be represented mathematically by the tensor product of the bra (ket) vectors constructed in each RHS.
%


We now establish the $\eta_1\otimes\eta_2$-RHS from the given RHSs, $\Phi_i\subset\mathcal{H}_i\subset \Phi_i^{\prime},\Phi_i^{\times}$ ($i=1,2$).
Let us suppose that each RHS possesses a positive operator $\eta_i$ defined on the Hilbert space $\mathcal{H}_i=(\mathcal{H}_i,\langle \cdot, \cdot \rangle_ {\mathcal{H}_i})$ such that $\eta_i$ is continuous on $\Phi_i=(\Phi_i,\tau_{\Phi_i})$ and satisfies $\eta_i\Phi_i\subset{\Phi_i}$.
In terms of the tensor product operator $\eta_1 \otimes \eta_2$ of $\eta_1$ and $\eta_2$, the following lemma is presented.

\begin{lemma}
\label{lemma2.1}
In the tensor product of the RHS (\ref{eqn:2-00}), the tensor product operator $\eta_1 \otimes \eta_2$ of $\eta_1$ and $\eta_2$ is a positive operator on $\mathcal{H}_1 \bar {\otimes} \mathcal{H}_2=(\mathcal{H}_1 \bar{\otimes} \mathcal{H}_2,\langle \cdot,\, \cdot\rangle_{\mathcal{H}_1 \bar{\otimes} \mathcal{H}_2})$ such that $\eta_1 \otimes \eta_2$ is continuous on $(\Phi_1 \hat{\otimes} {\Phi}_2,\widehat{\tau_p})$ and
$\eta_1\otimes \eta_2(\Phi_1 \hat{\otimes} {\Phi}_2)\subset (\Phi_1 \hat{\otimes} {\Phi}_2)$ holds.
\end{lemma}

\begin{proof}
From the positivity of $\eta_i$, it is apparent that each spectrum $Sp(\eta_i)$ of $\eta_i$ lies on the real half line $[0,+\infty)$.
Then $Sp(\eta_1\otimes\eta_2)
=Cl\{\lambda\mu ; \lambda\in Sp(\eta_1),\mu \in Sp(\eta_2)\}\subset Cl\{\lambda\mu;\lambda,\mu \in [0, +\infty)\}\subset [0,+\infty)$, where $Cl$ shows the closure of $[0,+\infty)$, which causes $\eta_1\otimes\eta_2$ to be positive.
We will show the continuity of $\eta_1\otimes\eta_2$ on $(\Phi_1 \hat{\otimes} {\Phi}_2,\widehat{\tau_p})$.
As each $\eta_i$ is continuous linear on $(\Phi_i,\tau_{\Phi_i})$ and satisfies $\eta_i\Phi_i\subset \Phi_i$, 
the restriction $\eta_i|\Phi_i : (\Phi_i,\tau_{\Phi_i})\to (\Phi_i,\tau_{\Phi_i})$ is continuous linear.
The extension theorem for locally convex spaces\cite{Jarchow} gives us the unique continuous linear extension of $\eta_1|\Phi_1 \otimes\eta_2|\Phi_2$ on $(\Phi_1 \hat{\otimes} {\Phi}_2,\widehat{\tau_p})$, 
$\eta_1|\Phi_1 \hat{\otimes}\eta_2|\Phi_2 : (\Phi_1 \hat{\otimes} {\Phi}_2,\widehat{\tau_p}) \to (\Phi_1 \hat{\otimes} {\Phi}_2,\widehat{\tau_p})$.
By the uniqueness, $\eta_1\otimes\eta_2|\Phi_1 \hat{\otimes} {\Phi}_2=\eta_1|\Phi_1 \hat{\otimes}\eta_2|\Phi_2$.
The proof is thus complete.
\end{proof}
\noindent
The positivity of the operator $\eta_1 \otimes \eta_2$ provides an RHS with a different metric from the RHS (\ref{eqn:2-00}), says the $\eta_1 \otimes \eta_2$-RHS\cite{Ohmori2022}.
This RHS can be expressed as the triplet of the following topological vector spaces:
\begin{eqnarray}
    \Phi_1 \hat{\otimes} {\Phi}_2 \subset 
    (
    \mathcal{H}_1 \widetilde{\bar {\otimes}} \mathcal{H}_2)_{\eta_1\otimes\eta_2}
    \subset 
    (\Phi_1 \hat{\otimes} {\Phi}_2)^{\prime},~(\Phi_1 \hat{\otimes} {\Phi}_2)^{\times},
    \label{eqn:2-2-1}
\end{eqnarray}
where $\Phi_1 \hat{\otimes} {\Phi}_2$, $(\Phi_1 \hat{\otimes} {\Phi}_2)^{\prime}$, and $(\Phi_1 \hat{\otimes} {\Phi}_2)^{\times}$ coincide with  (\ref{eqn:2-00}).
$(
    \mathcal{H}_1 \widetilde{\bar {\otimes}} \mathcal{H}_2)_{\eta_1\otimes\eta_2}$
represents the completion of the pre-Hilbert space 
$(
    \mathcal{H}_1  \bar{\otimes} \mathcal{H}_2)_{\eta_1\otimes\eta_2}
    =
    (\mathcal{H}_1 \bar{\otimes} \mathcal{H}_2,{\langle \cdot, \cdot \rangle_{\eta_1\otimes\eta_2}})$,
where ${\langle \cdot, \cdot \rangle_{\eta_1\otimes\eta_2}}$ is the inner product defined by
\begin{equation}
    \langle \phi, \psi \rangle_{\eta_1\otimes\eta_2}
    = \langle \phi,\, \eta_1\otimes\eta_2\psi\rangle_{\mathcal{H}_1 \bar{\otimes} \mathcal{H}_2} ~~~(\phi,\psi \in \mathcal{H}_1 \bar{\otimes} \mathcal{H}_2).
    \label{eqn:norm}
\end{equation}
Moreover, when the positive operator $\eta_1\otimes \eta_2$ on $\mathcal{H}_1 \bar{\otimes} \mathcal{H}_2$ is invertible, $\langle \cdot, \cdot \rangle_{\eta_1\otimes\eta_2}$ is equivalent to $\langle \cdot, \cdot \rangle_{\mathcal{H}_1 \bar{\otimes} \mathcal{H}_2}$, 
and (\ref{eqn:2-2-1}) is identified with 
\begin{eqnarray}
    \Phi_1 \hat{\otimes} {\Phi}_2 \subset 
    (
    \mathcal{H}_1 \bar {\otimes} \mathcal{H}_2)_{\eta_1\otimes\eta_2}
    \subset 
    (\Phi_1 \hat{\otimes} {\Phi}_2)^{\prime},~(\Phi_1 \hat{\otimes} {\Phi}_2)^{\times}.
    \label{eqn:eta_RHS}
\end{eqnarray}
Note that, if $\eta_1$ and $\eta_2$ are both invertible, then $\eta_1\otimes\eta_2$ becomes invertible; $(\eta_1\otimes\eta_2)^{-1}=\eta_1^{-1}\otimes\eta_2^{-1}$ holds.
Using $\eta_1 \otimes \eta_2$-RHS, the bra $\bra{\varphi}_{\eta_1\otimes\eta_2}$ and ket $\ket{\varphi}_{\eta_1\otimes\eta_2}$ vectors for each $\varphi \in \Phi_1 \hat{\otimes} {\Phi}_2$ can be defined as well as (\ref{eqn:bra-ket_definition}).
In particular, $\bra{\varphi}_{\eta_1\otimes\eta_2} (\phi)=\langle \eta_1\otimes\eta_2\varphi, \, \phi\rangle_{\mathcal{H}_1\bar{\otimes}\mathcal{H}_2}$ for any $\phi \in \Phi_1 \hat{\otimes} {\Phi}_2$,
and $\ket{\varphi}_{\eta_1\otimes\eta_2}$ is its complex conjugate.
Note that
\begin{eqnarray}
    \bra{\varphi}_{\eta_1\otimes\eta_2} =\bra{\eta_1\otimes\eta_2 \varphi}_{\mathcal{H}_1\bar{\otimes}\mathcal{H}_2}
    \textrm{ and }
    \ket{\varphi}_{\eta_1\otimes\eta_2} =\ket{\eta_1\otimes\eta_2 \varphi}_{\mathcal{H}_1\bar{\otimes}\mathcal{H}_2}
    \label{eqn:2-2-3}
\end{eqnarray}
are satisfied.

We can consider another construction of $\eta_1\otimes\eta_2$-RHS.
Let us set a pair of $\eta_i$-RHS, $\Phi_i \subset (\mathcal{H}_i)_{\eta_i} \subset \Phi_i^\prime,\Phi_i^\times$, characterized by the invertible positive operator $\eta_i$, where $(\mathcal{H}_i)_{\eta_i}=(\mathcal{H}_i, \langle \cdot, \cdot \rangle_{\mathcal{H}_i \eta_i})$ is the Hilbert space equipping the inner product $\langle \cdot, \cdot \rangle_{\mathcal{H}_i \eta_i}=\langle \cdot, \eta_i \cdot \rangle_{\mathcal{H}_i}$
($i=1,2$)\cite{Ohmori2022}.
From (\ref{eqn:2-00}), the tensor product is the RHS represented by
\begin{eqnarray}
    \Phi_1 \hat{\otimes} {\Phi}_2 
    \subset 
    (\mathcal{H}_1)_{\eta_1}\bar{\otimes}(\mathcal{H}_2)_{\eta_2}
    \subset 
    (\Phi_1 \hat{\otimes} {\Phi}_2)^{\prime},~(\Phi_1 \hat{\otimes} {\Phi}_2)^{\times}.
    \label{eqn:eta_RHS_2}
\end{eqnarray}
This RHS coincides exactly with (\ref{eqn:eta_RHS}).
In fact, the following theorem is obtained.
\begin{theorem}
\label{theorem2.2}
Let $(\mathcal{H}_1 \bar {\otimes} \mathcal{H}_2)_{\eta_1\otimes\eta_2}=(\mathcal{H}_1 \bar {\otimes} \mathcal{H}_2, \langle \cdot, \cdot \rangle_{\eta_1\otimes\eta_2})$ be the Hilbert space with the inner product (\ref{eqn:norm})  and let $(\mathcal{H}_1)_{\eta_1}$ and $(\mathcal{H}_2)_{\eta_2}$ be the Hilbert spaces induced by the positive invertible operators $\eta_1$ and $\eta_2$, respectively, where  $(\mathcal{H}_i)_{\eta_i}=(\mathcal{H}_i, \langle \cdot, \cdot \rangle_{\mathcal{H}_i \eta_i})~(i=1,2)$.
Then the tensor product, $(\mathcal{H}_1)_{\eta_1}\bar {\otimes}(\mathcal{H}_2)_{\eta_2}=(\mathcal{H}_1\bar{\otimes}\mathcal{H}_2,\langle \cdot, \cdot \rangle_{\mathcal{H}_1\eta_1\bar{\otimes} \mathcal{H}_2\eta_2})$,
of $(\mathcal{H}_1)_{\eta_1}$ and $(\mathcal{H}_2)_{\eta_2}$ coincides with $(\mathcal{H}_1 \bar {\otimes} \mathcal{H}_2)_{\eta_1\otimes\eta_2}$, namely, the relation 
\begin{eqnarray}
(\mathcal{H}_1)_{\eta_1}\bar {\otimes}(\mathcal{H}_2)_{\eta_2}
=(\mathcal{H}_1 \bar {\otimes} \mathcal{H}_2)_{\eta_1\otimes\eta_2}
    \label{eqn:relation2-2}
\end{eqnarray}
holds.
\end{theorem}

\begin{proof}
As each $\eta_i$ is a positive invertible operator, $\| \cdot \|_{\eta_i} \equiv \langle \cdot, \cdot \rangle_{\eta_i}^\frac{1}{2}$ is the equivalent norm with $\| \cdot \|_{\mathcal{H}_i} \equiv \langle \cdot, \cdot \rangle_{\mathcal{H}_i}^\frac{1}{2}$ in 
$\mathcal{H}_i$.
It follows that the algebraic tensor $\mathcal{H}_1\otimes \mathcal{H}_2$ possesses the equivalent norms 
$\|\cdot \|_{\mathcal{H}_1\eta_1\otimes \mathcal{H}_2\eta_2} \equiv \langle \cdot, \cdot \rangle_{\mathcal{H}_1\eta_1\otimes \mathcal{H}_2\eta_2}^\frac{1}{2}$ and $\|\cdot\| _{\mathcal{H}_1\otimes \mathcal{H}_2}  \equiv \langle \cdot, \cdot \rangle_{\mathcal{H}_1\otimes \mathcal{H}_2}^\frac{1}{2}$.
In addition, their completion $(\mathcal{H}_1)_{\eta_1}\bar {\otimes}(\mathcal{H}_2)_{\eta_2}=(\mathcal{H}_1\bar{\otimes}\mathcal{H}_2,\langle \cdot, \cdot \rangle_{\mathcal{H}_1\eta_1\bar{\otimes} \mathcal{H}_2\eta_2})$ and $(\mathcal{H}_1\bar{\otimes}\mathcal{H}_2,\langle \cdot, \cdot \rangle_{\mathcal{H}_1\bar{\otimes}\mathcal{H}_2})$ are also equivalent.
On the other hand, as $\eta_1\otimes\eta_2$
is the positive invertible operator, 
the norm $\|\cdot\|_{\eta_1\otimes\eta_2}\equiv \langle \cdot, \cdot \rangle_{\eta_1\otimes\eta_2}^\frac{1}{2}$ equivalent to  
$\|\cdot\|_{\mathcal{H}_1\bar{\otimes}\mathcal{H}_2}$ in $\mathcal{H}_1\bar{\otimes}\mathcal{H}_2$.
Thus, $\|\cdot\|_{\eta_1\otimes\eta_2}$ and $\|\cdot\|_{\mathcal{H}_1\eta_1\bar{\otimes} \mathcal{H}_2\eta_2}$ are equivalent, 
which indicates that the identity map $i:\varphi \mapsto \varphi$ on $\mathcal{H}_1\bar{\otimes}\mathcal{H}_2$ is a homeomorphism between the topologies induced from $\|\cdot\|_{\eta_1\otimes\eta_2}$ and $\|\cdot\|_{\mathcal{H}_1\eta_1\bar{\otimes} \mathcal{H}_2\eta_2}$. 
Let $\varphi\in (\mathcal{H}_1 \bar {\otimes} \mathcal{H}_2)_{\eta_1\otimes\eta_2}$; 
$\varphi$ is represented by $\varphi=\sum_{i=1}^\infty\varphi_i^1\otimes\varphi_i^2$ where $\{\varphi_i^1\}\subset \mathcal{H}_1$, $\{\varphi_i^2\}\subset \mathcal{H}_2$, 
and 
$\sum_{i=1}^\infty$ converges with respect to the norm $\|\cdot\|_{\eta_1\otimes\eta_2}$. 
As the identity map is homeomorphic, 
$\varphi=i(\varphi)=\sum_{i=1}^\infty\varphi_i^1\otimes\varphi_i^2$ where the sum $\sum_{i=1}^\infty$ converging with respect to $\|\cdot\|_{\mathcal{H}_1\eta_1\bar{\otimes} \mathcal{H}_2\eta_2}$.
Therefore, we obtain 
$\langle \varphi, \psi \rangle_{\eta_1\otimes\eta_2}=\langle \varphi, \psi \rangle_{\mathcal{H}_1\eta_1\bar{\otimes} \mathcal{H}_2\eta_2}$ for any $\varphi,\psi\in \mathcal{H}_1\bar{\otimes}\mathcal{H}_2$.
\end{proof}
\noindent
This theorem is utilized for discussing the spectral expansion of the quasi-Hermitian operator in a composite system in Sec. 3.

\subsection{Extension}
\label{sec:2.2}

For a given RHS, $\Phi\subset\mathcal{H}\subset\Phi^{\prime},\Phi^{\times}$, any linear operator $A:\mathcal{D}(A)\to \mathcal{H}$ on its domain $\mathcal{D}$ with the condition wherein $A$ is continuous on $\Phi$ and satisfies $A\Phi \subset \Phi$ can be always extended to the dual space $\Phi^{\prime}$ and the anti-dual space $\Phi^{\times}$\cite{Madrid2005}.
Such extensions $\hat{A}^{\prime}:\Phi' \rightarrow \Phi'$ and $\hat{A}^{\times}:\Phi^{\times} \rightarrow \Phi^{\times}$ are defined by
$(\hat {A}^j (f))(\phi):=f(A(\phi))$~$(j=\prime,\times)$,
for any $\phi\in\Phi$ and $f\in \Phi^j$.
Moreover, noting $\Phi^{\prime} \cap \Phi^{\times}=\{\hat{0}\}$ ($\hat{0}$ represents the zero-valued functional on $\Phi$),
these mappings can be combined into one operator $\hat{A} : \Phi^{\prime} \cup \Phi^{\times}\to\Phi^{\prime} \cup \Phi^{\times}$,
where
\begin{eqnarray}
(\hat{A} (f))(\phi)=f(A (\phi))
    \label{eqn:extension_definition}
\end{eqnarray}
for $f\in \Phi^{\prime} \cup \Phi^{\times}$ and $\phi\in\Phi$.
When $f\in \Phi^{j}$, we obtain $\hat{A} (f) \in \Phi^{j}$ ($j=\prime,\times$), and the relation
$\hat{A} (f^*)={\hat{A} ({f})}^*$ is satisfied.

This procedure of construction of the extension (\ref{eqn:extension_definition}) can be immediately applied to the positive operator $\eta_1 \otimes \eta_2$ in  (\ref{eqn:2-2-1}) by Lemma \ref{lemma2.1}.
The extension of $\eta_1\otimes\eta_2$ to $(\Phi_1 \hat{\otimes} {\Phi}_2)^{\prime}$ and $(\Phi_1 \hat{\otimes} {\Phi}_2)^{\times}$ is denoted by $\widehat{\eta_1\otimes\eta_2} : (\Phi_1 \hat{\otimes} {\Phi}_2)^{\prime}\cup(\Phi_1 \hat{\otimes} {\Phi}_2)^{\times}\to (\Phi_1 \hat{\otimes} {\Phi}_2)^{\prime}\cup(\Phi_1 \hat{\otimes} {\Phi}_2)^{\times},
$
where $\widehat{\eta_1\otimes\eta_2}$ is given by
\begin{eqnarray}
    (\widehat{\eta_1\otimes\eta_2} (f))(\phi)=f({\eta_1\otimes\eta_2} (\phi))
    \label{eqn:2-ext1}
\end{eqnarray}
for $f\in (\Phi_1 \hat{\otimes} {\Phi}_2)^{\prime}\cup(\Phi_1 \hat{\otimes} {\Phi}_2)^{\times}$ and $\phi\in \Phi_1\hat{\otimes}\Phi_2$.
This satisfies the relations, $\widehat{\eta_1\otimes\eta_2}(\Phi_1 \hat{\otimes} {\Phi}_2)^{j}\subset (\Phi_1 \hat{\otimes} {\Phi}_2)^{j}$~$(j=\prime,\times)$,  
$\bra{\varphi}_{\mathcal{H}_1\bar{\otimes}\mathcal{H}_2}\in (\Phi_1 \hat{\otimes} {\Phi}_2)^{\prime}$,
$\ket{\varphi}_{\mathcal{H}_1\bar{\otimes}\mathcal{H}_2}\in (\Phi_1 \hat{\otimes} {\Phi}_2)^{\times}$, 
and $\widehat{\eta_1\otimes\eta_2}(f^*)=\widehat{\eta_1\otimes\eta_2}(f)^*$.
Furthermore,
\begin{eqnarray}
    \bra{\varphi}_{\mathcal{H}_1\bar{\otimes}\mathcal{H}_2}\widehat{\eta_1\otimes\eta_2}=\bra{\varphi}_{\eta_1\otimes\eta_2},
~~    \widehat{\eta_1\otimes\eta_2}\ket{\varphi}_{\mathcal{H}_1\bar{\otimes}\mathcal{H}_2}=\ket{\varphi}_{\eta_1\otimes\eta_2}
    \label{eqn:2-ext2}
\end{eqnarray}
are confirmed from (\ref{eqn:2-2-3}), 
where $\bra{\varphi}_{\mathcal{H}_1\bar{\otimes}\mathcal{H}_2}\widehat{\eta_1\otimes\eta_2}$ denotes
$\widehat{\eta_1\otimes\eta_2}(\bra{\varphi}_{\mathcal{H}_1\bar{\otimes}\mathcal{H}_2})$.
The relations in (\ref{eqn:2-ext2}) present the transformation between the bra and ket vectors of $\mathcal{H}_1\bar{\otimes}\mathcal{H}_2$-system into those of the $(\mathcal{H}_1 \widetilde{\bar {\otimes}} \mathcal{H}_2)_{\eta_1\otimes\eta_2}$-system.
In particular, for $\varphi_1\otimes\varphi_2\in\Phi_1\otimes\Phi_2$, we have
\begin{eqnarray}
    \bra{\varphi_1\otimes\varphi_2}_{\mathcal{H}_1\bar{\otimes}\mathcal{H}_2}\widehat{\eta_1\otimes\eta_2} 
    & = & \bra{\varphi_1\otimes\varphi_2}_{\eta_1 \otimes \eta_2}
     =  \bra{\eta_1\varphi_1\otimes\eta_2\varphi_2}_{\mathcal{H}_1\bar{\otimes}\mathcal{H}_2}
     \nonumber\\
     & = &
     \bra{\eta_1\varphi_1}_{\mathcal{H}_1}\otimes\bra{\eta_2\varphi_2}_{\mathcal{H}_2}
      = 
     \bra{\varphi_1}_{\eta_1}\otimes\bra{\varphi_2}_{\eta_2}
     \nonumber\\
     & = &
     \bra{\varphi_1}_{\mathcal{H}_1}\hat{\eta}_1\otimes\bra{\varphi_2}_{\mathcal{H}_2}\hat{\eta}_2,
    \label{eqn:2-ext3a}
\end{eqnarray}
\begin{eqnarray}
    \widehat{\eta_1\otimes\eta_2} 
    \ket{\varphi_1\otimes\varphi_2}_{\mathcal{H}_1\bar{\otimes}\mathcal{H}_2}
    & = & \ket{\varphi_1\otimes\varphi_2}_{\eta_1 \otimes \eta_2}
     =  \ket{\eta_1\varphi_1\otimes\eta_2\varphi_2}_{\mathcal{H}_1\bar{\otimes}\mathcal{H}_2}
     \nonumber\\
     & = &
     \ket{\eta_1\varphi_1}_{\mathcal{H}_1}\otimes\ket{\eta_2\varphi_2}_{\mathcal{H}_2}
     = 
     \ket{\varphi_1}_{\eta_1}\otimes\ket{\varphi_2}_{\eta_2}
    \nonumber\\
     & = &
     \hat{\eta}_1\ket{\varphi_1}_{\mathcal{H}_1}\otimes\hat{\eta}_2\ket{\varphi_2}_{\mathcal{H}_2},
     \label{eqn:2-ext3b}
\end{eqnarray}
where $\hat{\eta}_i$ is the extension of $\eta_i$ to $\Phi_i^\prime\cup\Phi_i^\times$, which is constructed using (\ref{eqn:extension_definition}).
When each $\eta_i$ is invertible, 
the inverse operator $\widehat{\eta_1\otimes\eta_2}^{-1}$ of $\widehat{\eta_1\otimes\eta_2}$ can be defined satisfying $\widehat{\eta_1\otimes\eta_2}^{-1}=\widehat{(\eta_1\otimes\eta_2)^{-1}}$ 
where $
[\widehat{(\eta_1\otimes\eta_2)^{-1}} (f)](\phi)=f((\eta_1\otimes\eta_2)^{-1} (\phi))
$
for $f\in (\Phi_1 \hat{\otimes} {\Phi}_2)^{\prime}\cup(\Phi_1 \hat{\otimes} {\Phi}_2)^{\times}$ and $\phi\in \Phi_1\hat{\otimes}\Phi_2$.
Then, the inverse relations of (\ref{eqn:2-2-3}) and (\ref{eqn:2-ext2}) are as follows:
\begin{eqnarray}
    \bra{\varphi}_{\mathcal{H}_1 \bar{\otimes} \mathcal{H}_2}
    & = &
    \bra{(\eta_1\otimes\eta_2)^{-1}\varphi}_{\eta_1{\otimes}\eta_2}
    =
    \bra{\varphi}_{\eta_1{\otimes}\eta_2}\widehat{\eta_1\otimes\eta_2}^{-1}, 
    \label{eqn:2-ext40a}
    \\
    \ket{\varphi}_{\mathcal{H}_1 \bar{\otimes} \mathcal{H}_2}
    & = &
    \ket{(\eta_1\otimes\eta_2)^{-1}\varphi}_{\eta_1{\otimes}\eta_2}
    =
    \widehat{\eta_1\otimes\eta_2}^{-1}\ket{\varphi}_{\eta_1{\otimes}\eta_2}.
     \label{eqn:2-ext40b}
\end{eqnarray}
Similarly, for $\varphi_1\otimes\varphi_2\in\Phi_1\otimes\Phi_2$, we have the inverse relations of (\ref{eqn:2-ext3a}) and (\ref{eqn:2-ext3b}),
\begin{eqnarray}
    \bra{\varphi_1\otimes\varphi_2}_{\eta_1{\otimes}\eta_2}\widehat{\eta_1\otimes\eta_2}^{-1} 
    & = & \bra{\varphi_1\otimes\varphi_2}_{\mathcal{H}_1 \bar{\otimes} \mathcal{H}_2}
     =  
     \bra{\varphi_1}_{\mathcal{H}_1}\hat{\eta}_1^{-1}\otimes     \bra{\varphi_2}_{\mathcal{H}_2}\hat{\eta}_2^{-1},
    \label{eqn:2-ext4a}
    \\
    \widehat{\eta_1\otimes\eta_2}^{-1}\ket{\varphi_1\otimes\varphi_2}_{\eta_1{\otimes}\eta_2} 
    & = & \ket{\varphi_1\otimes\varphi_2}_{\mathcal{H}_1 \bar{\otimes} \mathcal{H}_2}
     =  
     \hat{\eta}_1^{-1}\ket{\varphi_1}_{\eta_1}\otimes\hat{\eta}_2^{-1}\ket{\varphi_2}_{\eta_2},
     \label{eqn:2-ext4b}
\end{eqnarray}
where each $\hat{\eta_i}^{-1}$ is the inverse of $\hat{\eta}_i$ on $\Phi_i^\prime\cup\Phi_i^\times$.

%

\section{Spectral expansion of the quasi-Hermitian composite operator}
\label{sec:3}

\subsection{Set up}
\label{sec:3.1}

Let $\mathcal{H}_i=(\mathcal{H}_i,\langle \cdot, \cdot \rangle_{\mathcal{H}_i})$ be a Hilbert space, and let $\eta_i$ be a positive invertible operator on $\mathcal{H}_i$ ($i=1,2$). 
For each $i$, we set an $\eta_i$-quasi Hermitian operator $A_i : \mathcal{D}(A_i)\to \mathcal{H}_i$ on its domain $\mathcal{D}(A_i)$ into a Hilbert space $\mathcal{H}_i$,
where $A_i$ has the symmetric structure for its adjoint $A_i^{\dagger}$\cite{Dieudonne1961,Mos2010,Antoine2013},
\begin{eqnarray}
    A_i^{\dagger}= {\eta_i}
    A_i{\eta_i}^{-1}~~(i=1,2).
    \label{eqn:3-00}
\end{eqnarray}
Here, $\eta_i$ becomes an intertwining operator for $A_i$ and $A_i^\dagger$.
Under the Hilbert space $(\mathcal{H}_i)_{\eta_i}=(\mathcal{H}_i, \langle \cdot, \cdot \rangle_{\mathcal{H}_i \eta_i})$ the $\eta_i$-quasi Hermitian operator is an Hermitian (self-adjoint) operator\cite{Antoine2013}.
%

%
Considering the Hamiltonian of a composite system\cite{Simon1980,Messiah},
we now handle the following type of a linear operator in the tensor product $\mathcal{H}_1 \bar{\otimes}\mathcal{H}_2=(\mathcal{H}_1 \bar{\otimes}\mathcal{H}_2,\langle \cdot, \cdot \rangle_{\mathcal{H}_1 \bar{\otimes}\mathcal{H}_2})$,
\begin{eqnarray}
    A=\overline{A_1\otimes I_2 +I_1\otimes A_2} : \mathcal{D}({A}) \to \mathcal{H}_1 \bar{\otimes}\mathcal{H}_2, 
    \label{eqn:composite_operator}
\end{eqnarray}
 where $A$ is the closure of the operator $A_1\otimes I_2 +I_1\otimes A_2$ in $\mathcal{H}_1 \bar{\otimes}\mathcal{H}_2$ and $I_i$ is the identity map for $\mathcal{H}_i$ $(i=1,2)$.
Note that $A$ is {\it not} self-adjoint in $\mathcal{H}_1 \bar{\otimes}\mathcal{H}_2$ as each $A_i$ is not self-adjoint in $\mathcal{H}_i$.
However, it becomes a self-adjoint operator in $(\mathcal{H}_1 \bar{\otimes}\mathcal{H}_2)_{\eta_1\otimes\eta_2}=(\mathcal{H}_1 \bar {\otimes} \mathcal{H}_2, \langle \cdot, \cdot \rangle_{\eta_1\otimes\eta_2})$.
To demonstrate this, the following lemma is derived. 
\begin{lemma}
\label{lemma3.1}
Let $A:\mathcal{D}(A)\to \mathcal{H}$ be a closable linear operator in a Hilbert space $(\mathcal{H},\langle \cdot, \cdot \rangle_1)$, the closure of which is denoted by $\bar{A}$.
If the norm induced from the inner product $\langle \cdot, \cdot \rangle_2$ of $\mathcal{H}$ is equivalent to that induced from $\langle \cdot, \cdot \rangle_1$, then $A$ is closable with respect to $(\mathcal{H},\langle \cdot, \cdot \rangle_2)$ and its closure coincides with $\bar{A}$.
\end{lemma}
\begin{proof}
Suppose that a sequence $\{\psi_n\}\subset\mathcal{D}(A)$ converges to zero with respect to $\langle \cdot, \cdot \rangle_2$ and $A\psi_n$ converges to $\phi$
with respect to $\langle \cdot, \cdot \rangle_2$, i.e.,
$\|\psi_n\|_2\to 0$ as $n\to \infty$ and 
$\|A\psi_n-\phi\|_2\to 0$ as $n\to \infty$, where $\|\cdot \|_2=\sqrt{\langle \cdot, \cdot \rangle_2}$.
Since $\|\cdot \|_2$ is equivalent to $\|\cdot \|_1=\sqrt{\langle \cdot, \cdot \rangle_1}$, there exists $C>0$ such that $C\| \psi_n \|_2\geq \|\psi_n\|_1$ for all $n$, 
and then $\|\psi_n\|_1\to 0$ and $\|A\psi_n-\phi\|_1\to 0$.
As $A$ is closable in $(\mathcal{H},\langle \cdot, \cdot \rangle_1)$, $\phi=0$, which shows that $A$ is also closable with respect to $\langle \cdot, \cdot \rangle_2$.

Let $\bar{A}^{\prime}$ be the closure of $A$ with respect to $\langle \cdot, \cdot \rangle_2$.
From the equivalence of the norms $\| \cdot \|_1$ and $\| \cdot \|_2$, it is easy to show that $\bar{A}^{\prime}$ is a closed map of $A$ with respect to $\langle \cdot, \cdot \rangle_1$. 
Owing to the minimality of closure, $\bar{A}\subset \bar{A}^{\prime}$.
Similarly, we can obtain $\bar{A}^{\prime}\subset \bar{A}$.
\end{proof}
\begin{theorem}
\label{theorem3.2}
    The linear operator $A:\mathcal{D}(A)\to \mathcal{H}_1 \bar{\otimes}\mathcal{H}_2$ given by (\ref{eqn:composite_operator}) is self-adjoint with respect to $(\mathcal{H}_1 \bar{\otimes}\mathcal{H}_2)_{\eta_1\otimes\eta_2}$.    
\end{theorem}
\begin{proof}
    
From Lemma \ref{lemma3.1}, 
the closure of the operator $A_1\otimes I_2 +I_1\otimes A_2$ in $(\mathcal{H}_1 \bar{\otimes}\mathcal{H}_2)_{\eta_1\otimes\eta_2}$ coincides with of in $\mathcal{H}_1 \bar{\otimes}\mathcal{H}_2$.
Thus, $A$ is the closure in $(\mathcal{H}_1 \bar{\otimes}\mathcal{H}_2)_{\eta_1\otimes\eta_2}$; 
the operator $A$ given by (\ref{eqn:composite_operator}) is well-defined in $(\mathcal{H}_1 \bar{\otimes}\mathcal{H}_2)_{\eta_1\otimes\eta_2}$.
We now recall that 
for Hilbert spaces $\mathcal{H}_1$ and $\mathcal{H}_2$ and self-adjoint operators $A_1:\mathcal{D}(A_1)\to \mathcal{H}_1$ in $\mathcal{H}_1$ and $A_2:\mathcal{D}(A_2)\to \mathcal{H}_2$  in $\mathcal{H}_2$, the linear operator  $T=A_1\otimes I_2 +I_1\otimes A_2:\mathcal{D}(T)\to \mathcal{H}_1\bar{\otimes}\mathcal{H}_2$ becomes an essential self-adjoint operator and hence its closure $\bar{T}$ is self-adjoint\cite{Simon1980}.
As each $\eta_i$-quasi Hermitian operator $A_i$ is self-adjoint in $(\mathcal{H}_i)_{\eta_i}$, $T=A_1\otimes I_2 +I_1\otimes A_2:\mathcal{D}(T)\to \mathcal{H}_1\bar{\otimes}\mathcal{H}_2$ is an essential self-adjoint operator in $(\mathcal{H}_1)_{\eta_1}\bar {\otimes}(\mathcal{H}_2)_{\eta_2}$; the closure $\bar{T}$ is self-adjoint.
By Lemma \ref{lemma2.1}, 
$(\mathcal{H}_1)_{\eta_1}\bar {\otimes}(\mathcal{H}_2)_{\eta_2}=(\mathcal{H}_1\bar{\otimes}\mathcal{H}_2)_{\eta_1\bar{\otimes}\eta_2}$, which shows that $A=\bar{T}$ is a self-adjoint operator in $(\mathcal{H}_1\bar{\otimes}\mathcal{H}_2)_{\eta_1\bar{\otimes}\eta_2}$.
\end{proof}
\begin{corollary}
    \label{corollary}
    For an $\eta_i$-quasi Hermitian operator, $A_i$, satisfying (\ref{eqn:3-00}) ($i=1,2$), the linear operator $A$ of the form presented in (\ref{eqn:composite_operator}) is an $\eta_1\otimes\eta_2$-quasi Hermitian operator satisfying the symmetric structure,
    \begin{eqnarray}
    A^{\dagger}= {(\eta_1\otimes\eta_2)}
    A({\eta_1\otimes\eta_2})^{-1}.
    \label{eqn:3-01}
    \end{eqnarray}
\end{corollary}
\begin{proof}
    It is apparent from Theorem \ref{theorem3.2}.
\end{proof}
\noindent
From these results, the linear operator of the form presented in (\ref{eqn:composite_operator}) represents an $\eta_1\otimes\eta_2$-quasi Hermitian operator with respect to the RHS in (\ref{eqn:2-00}).

\subsection{Spectral expansions}
\label{sec:3.2}

Let us now focus on the spectral expansions of an $\eta_1\otimes\eta_2$-quasi Hermitian as given in (\ref{eqn:composite_operator}).
From Theorem \ref{theorem2.2}, it is found that the $\eta_1\otimes\eta_2$-RHS is the tensor product of $\eta$-RHSs.
Thus, the spectral expansion of $\eta_1\otimes\eta_2$-quasi Hermitian can be obtained using the general results for spectral expansions in the tensor product of RHS\cite{Ohmori2024}.
On setting the $\eta_i$-RHS, $\Phi_i\subset (\mathcal{H}_i)_{\eta_i} \subset\Phi_i^{\prime},\Phi_i^{\times}$ $(i=1,2)$,
we assume that each $\eta_i$-quasi Hermitian $A_i$ is continuous on $(\Phi_i,\tau_{\Phi_i})$ and $A_i \Phi_i\subset\Phi_i$.
Then, for $\eta_1\otimes\eta_2$-RHS (\ref{eqn:eta_RHS}),
$A$ has the spectrum $Sp(A)=Cl(Sp(A_1)+Sp(A_2))$ lying on the real line
($ClX$ is the closure of a set $X$ in the real line)
because of its self-adjointness, 
and $A$ satisfies the continuity in $\Phi_1\hat {\otimes}\Phi_2$ and $A(\Phi_1\hat {\otimes}\Phi_2)\subset \Phi_1\hat {\otimes}\Phi_2$.
Therefore, with the aid of the nuclear spectral theorem, under $\eta_1\otimes\eta_2$-RHS, we have the following relations for any $\varphi, \psi \in \Phi_1 \widehat{\otimes} \Phi_2$:
\begin{eqnarray}
    \langle \varphi,\, \psi\rangle_{\eta_1{\otimes} \eta_2}
    & = & 
     \displaystyle
     \int_{\lambda\in Sp({A})}
     \braket{\hat{\varphi}}{\hat{\psi}}_\lambda
     d\mu_\lambda,
    \label{eqn:o3-2-1a}\\
    \langle \varphi,\, A\psi\rangle_{\eta_1{\otimes} \eta_2}
    & = & 
     \displaystyle
     \int_{\lambda\in Sp({A})}
     \lambda\braket{\hat{\varphi}}{\hat{\psi}}_\lambda
     d\mu_\lambda,
    \label{eqn:o3-2-1b}
\end{eqnarray}
with
\begin{eqnarray}
    \displaystyle\braket{\hat{\varphi}}{\hat{\psi}}_\lambda
    =
    \displaystyle\int_{\lambda=\lambda_1+\lambda_2}
     \sum_{k=1}^{dim \hat{\mathcal{H}_1}(\lambda_1)} 
     \sum_{l=1}^{dim \hat{\mathcal{H}_2}(\lambda_2)}
     (\ket{\lambda_1}_{\eta_1,k}\otimes
     \ket{\lambda_2}_{\eta_2,l})(\varphi)
     (\bra{\lambda_1}_{\eta_1,k}\otimes
     \bra{\lambda_2}_{\eta_2,l})(\psi)
     d\sigma^\lambda_{\lambda_1,\lambda_2},
    \label{eqn:o3-2-1c}
\end{eqnarray}
where $\mu_\lambda$ is the Borel measure given in the Neumanns' complete spectral theorem\cite{Maurin1968}, and 
$\sigma^\lambda_{\lambda_1,\lambda_2}$ is also a Borel measure on $\mathbb{R}^2$ whose support is contained in the set $\{(\lambda_1,\lambda_2)\in \mathbb{R}^2 ; \lambda=\lambda_1+\lambda_2, \lambda_i\in Sp(A_i) (i=1,2)\}$.
The notations of bra, $\bra{\lambda_1}_{\eta_1,k}(k=1,2,\cdots, dim \hat{\mathcal{H}_1}(\lambda_1))$ and $\bra{\lambda_2}_{\eta_2,l}~(l=1,2,\cdots, dim \hat{\mathcal{H}_2}(\lambda_2))$, represent the generalized eigenvectors of $A_1$ and $A_2$ in $\Phi_1^\prime$ and $\Phi_2^{\prime}$, corresponding to $\lambda_1$ and $\lambda_2$, which are obtained from the nuclear spectral theorem, respectively. 
The ket $\ket{\lambda}$ represents the complex conjugate of the bra $\bra{\lambda}$, 
i.e., $\ket{\lambda}(\varphi)=\bra{\lambda}^*(\varphi)=\bra{\lambda}(\varphi)^*$. 
When $dim \hat{\mathcal{H}_1}(\lambda_1)=dim \hat{\mathcal{H}_2}(\lambda_2)=1$,
the relations in ($\ref{eqn:o3-2-1a}$) and ($\ref{eqn:o3-2-1b}$) become
\begin{eqnarray}
    \langle \varphi,\, \psi\rangle_{\eta_1{\otimes} \eta_2}
    & = & 
     \displaystyle
     \int\limits_{\lambda\in Sp({A})}
    \Big{\{}
    \displaystyle\int\limits_{\lambda=\lambda_1+\lambda_2}
     (\ket{\lambda_1}_{\eta_1}\otimes
     \ket{\lambda_2}_{\eta_2})(\varphi)
     (\bra{\lambda_1}_{\eta_1}\otimes
     \bra{\lambda_2}_{\eta_2})(\psi)
     d\sigma^\lambda_{\lambda_1,\lambda_2}
     \Big{\}}
     d\mu_\lambda,
    \label{eqn:o3-2-2a}
    \\
    \langle \varphi,\, A\psi\rangle_{\eta_1{\otimes} \eta_2}
    & = & 
     \displaystyle
     \int\limits_{\lambda\in Sp({A})}
    \lambda
    \Big{\{}
    \displaystyle\int\limits_{\lambda=\lambda_1+\lambda_2}
     (\ket{\lambda_1}_{\eta_1}\otimes
     \ket{\lambda_2}_{\eta_2})(\varphi)
     (\bra{\lambda_1}_{\eta_1}\otimes
     \bra{\lambda_2}_{\eta_2})(\psi)
     d\sigma^\lambda_{\lambda_1,\lambda_2}
     \Big{\}}
     d\mu_\lambda,
    \label{eqn:o3-2-2b}
\end{eqnarray}
for any $\varphi, \psi \in \Phi_1 \widehat{\otimes} \Phi_2$.
The bra vector $\bra{\lambda_1}_{\eta_1}\otimes
     \bra{\lambda_2}_{\eta_2}$ and ket vector $\ket{\lambda_1}_{\eta_1}\otimes
     \ket{\lambda_2}_{\eta_2}$ 
belong to $(\Phi_1 \widehat{\otimes} \Phi_2)^\prime$ and 
     $(\Phi_1 \widehat{\otimes} \Phi_2)^{\times}$, respectively,
and as they are generalized eigenvectors of $A$, the eigenequations
\begin{eqnarray}
    \bra{\lambda_1}_{\eta_1}\otimes
     \bra{\lambda_2}_{\eta_2}(A\varphi)
     & = &
     (\lambda_1+\lambda_2)\bra{\lambda_1}_{\eta_1}\otimes
     \bra{\lambda_2}_{\eta_2}(\varphi),
     \label{eqn:eigenequation_tensors1}\\
     \ket{\lambda_1}_{\eta_1}\otimes
     \ket{\lambda_2}_{\eta_2}(A\varphi)
     & = &
     (\lambda_1+\lambda_2)\ket{\lambda_1}_{\eta_1}\otimes
     \ket{\lambda_2}_{\eta_2}(\varphi),
    \label{eqn:eigenequation_tensors2}
\end{eqnarray}
are satisfied.
We denote $\bra{\lambda_1}_{\eta_1}\otimes
     \bra{\lambda_2}_{\eta_2}(\varphi)$ and $\bra{\lambda_1}_{\eta_1}\otimes
     \bra{\lambda_2}_{\eta_2}$ by 
     $\bra{\lambda_1}_{\eta_1}\otimes\bra{\lambda_2}_{\eta_2}\ket{\varphi}_{\mathcal{H}_1\bar{\otimes}\mathcal{H}_2}
     $ = 
     $\braket{\lambda_1\otimes\lambda_2}{\varphi}_{\eta_1\otimes\eta_2}
     $ and $\bra{\lambda_1\otimes\lambda_2}_{\eta_1\otimes\eta_2}$, respectively.
Similarly, $\ket{\lambda_1}_{\eta_1}\otimes
     \ket{\lambda_2}_{\eta_2}(\varphi)$ and $\ket{\lambda_1}_{\eta_1}\otimes
     \ket{\lambda_2}_{\eta_2}$ are denoted by
     $\bra{\varphi}_{\mathcal{H}_1\bar{\otimes}\mathcal{H}_2}\ket{\lambda_1}_{\eta_1}\otimes\ket{\lambda_2}_{\eta_2}
     $ = 
     $\braket{\varphi}{\lambda_1\otimes\lambda_2}_{\eta_1\otimes\eta_2}
     $ and $\ket{\lambda_1\otimes\lambda_2}_{\eta_1\otimes\eta_2}$, respectively.
In addition, we adopt the following abbreviation\cite{Ohmori2024}:
\begin{eqnarray}
    \int_{\lambda\in Sp(A)}\int_{\lambda=\lambda_1+\lambda_2}\to \int_{Sp(A)}~~ 
    \text{and}~~
    d\sigma^\lambda_{\lambda_1,\lambda_2}d\mu_\lambda \to d\nu.
    \label{abbreviation}
\end{eqnarray}
Considering these notation, the spectral expansions of the bra and ket vectors for the $\eta_1\otimes\eta_2$-quasi Hermitian $A$ are derived from (\ref{eqn:o3-2-2a}) and (\ref{eqn:o3-2-2b}), as follows: for $\varphi \in \Phi_1 \widehat{\otimes} \Phi_2$,
    \begin{eqnarray}
    \bra{\varphi}_{\eta_1{\otimes} \eta_2}
    & = &
    \displaystyle\int_{Sp(A)}
     \braket{\varphi}{\lambda_1\otimes
     \lambda_2}_{\eta_1{\otimes} \eta_2}
     \bra{\lambda_1\otimes
     \lambda_2}_{\eta_1{\otimes} \eta_2} 
    d\nu,\\
    \label{spectralexpansion_bra_aa}
    \bra{A\varphi}_{\eta_1{\otimes} \eta_2}
    & = &
    \displaystyle\int_{Sp(A)}\lambda
     \braket{\varphi}{\lambda_1\otimes
     \lambda_2}_{\eta_1{\otimes} \eta_2}
     \bra{\lambda_1\otimes
     \lambda_2}_{\eta_1{\otimes} \eta_2} 
    d\nu,
    \label{spectralexpansion_bra_ab}
\end{eqnarray}
and 
\begin{eqnarray}
     \ket{\varphi}_{\eta_1{\otimes} \eta_2}
    & = &
    \displaystyle\int_{Sp(A)}
     \braket{\lambda_1\otimes
     \lambda_2}{\varphi}_{\eta_1{\otimes} \eta_2}    
     \ket{\lambda_1\otimes
     \lambda_2}_{\eta_1{\otimes} \eta_2}
    d\nu,\\
    \label{spectralexpansion_ket_aa}
         \ket{A\varphi}_{\eta_1{\otimes} \eta_2}
    & = & 
    \displaystyle\int_{Sp(A)}\lambda
     \braket{\lambda_1\otimes
     \lambda_2}{\varphi}_{\eta_1{\otimes} \eta_2}    \ket{\lambda_1\otimes
     \lambda_2}_{\eta_1{\otimes} \eta_2}
    d\nu.
    \label{spectralexpansion_ket_ab}
    \end{eqnarray}
These spectral expansions can be converted into that in the $\mathcal{H}_1\bar{\otimes}\mathcal{H}_2$-system;
using (\ref{eqn:2-ext40a}) and (\ref{eqn:2-ext40b}), we obtain
\begin{eqnarray}
     \ket{\varphi}_{\mathcal{H}_1\bar{\otimes}\mathcal{H}_2}
     & = & \ket{(\eta_1\otimes\eta_2)^{-1}\varphi}_{\eta_1{\otimes} \eta_2}
     \nonumber\\ 
      & = &
    \displaystyle\int_{Sp(A)}
     \bra{\lambda_1}_{\eta_1}\otimes
     \bra{\lambda_2}_{\eta_2} ((\eta_1\otimes\eta_2)^{-1}\varphi)    
     \ket{\lambda_1}_{\eta_2}\otimes
     \ket{\lambda_2}_{\eta_2}
    d\nu
    \nonumber\\
        & = &
    \displaystyle\int_{Sp(A)}
     \widehat{\eta_1\otimes\eta_2}^{-1}(\bra{\lambda_1}_{\eta_1}\otimes
     \bra{\lambda_2}_{\eta_2}(\varphi))    
     \ket{\lambda_1}_{\eta_1}\otimes
     \ket{\lambda_2}_{\eta_2}
    d\nu
    \nonumber\\
        & = &
    \displaystyle\int_{Sp(A)}
     \widehat{\eta_1\otimes\eta_2}^{-1}(\bra{\lambda_1}_{\eta_1}\otimes
     \bra{\lambda_2}_{\eta_2}\ket{\varphi}_{\mathcal{H}_1\bar{\otimes}\mathcal{H}_2})    
     \ket{\lambda_1}_{\eta_1}\otimes
     \ket{\lambda_2}_{\eta_2}
    d\nu
    \nonumber\\
    & = &
    \displaystyle\int_{Sp(A)}
     \bra{\lambda_1}_{\eta_1}\otimes
     \bra{\lambda_2}_{\eta_2}\widehat{\eta_1\otimes\eta_2}^{-1} \ket{\varphi}_{\mathcal{H}_1\bar{\otimes} \mathcal{H}_2}    
     \ket{\lambda_1}_{\eta_1}\otimes
     \ket{\lambda_2}_{\eta_2}
    d\nu.
    \label{spectralexpansion_ket_H_1}
    \end{eqnarray}
%
%
Similarly, 
\begin{eqnarray}
     \ket{A\varphi}_{\mathcal{H}_1\bar{\otimes}\mathcal{H}_2}
     & = & 
    \displaystyle\int_{Sp(A)}
     \lambda
     \bra{\lambda_1}_{\eta_1}\otimes
     \bra{\lambda_2}_{\eta_2}\ket{\varphi}_{\mathcal{H}_1\bar{\otimes} \mathcal{H}_2}    
     \widehat{\eta_1\otimes\eta_2}^{-1}
     \ket{\lambda_1}_{\eta_1}\otimes
     \ket{\lambda_2}_{\eta_2}
    d\nu,
    \label{spectralexpansion_ket_H_2}\\
     \bra{\varphi}_{\mathcal{H}_1\bar{\otimes}\mathcal{H}_2}
     & = & 
    \displaystyle\int_{Sp(A)}
     \bra{\varphi}_{\mathcal{H}_1\bar{\otimes} \mathcal{H}_2}\widehat{\eta_1\otimes\eta_2}^{-1}
         \ket{\lambda_1}_{\eta_1}\otimes
     \ket{\lambda_2}_{\eta_2}
     \bra{\lambda_1}_{\eta_1}\otimes
     \bra{\lambda_2}_{\eta_2}     
    d\nu,
    \label{spectralexpansion_bra_H_1}\\
     \bra{A\varphi}_{\mathcal{H}_1\bar{\otimes}\mathcal{H}_2}
     & = & 
    \displaystyle\int_{Sp(A)}
     \lambda
     \bra{\varphi}_{\mathcal{H}_1\bar{\otimes} \mathcal{H}_2}
         \ket{\lambda_1}_{\eta_1}\otimes
     \ket{\lambda_2}_{\eta_2}
     \bra{\lambda_1}_{\eta_1}\otimes
     \bra{\lambda_2}_{\eta_2}\widehat{\eta_1\otimes\eta_2}^{-1}
    d\nu.
    \label{spectralexpansion_bra_H_2}
    \end{eqnarray}
The spectral expansions concerning the adjoint operator $A^\dagger$ in $\mathcal{H}_1\bar{\otimes}\mathcal{H}_2$ can be also obtained by
\begin{eqnarray}
     \ket{A^\dagger \varphi}_{\mathcal{H}_1\bar{\otimes}\mathcal{H}_2}
     & = & 
    \displaystyle\int_{Sp(A)}
     \lambda
     \bra{\lambda_1}_{\eta_1}\otimes
     \bra{\lambda_2}_{\eta_2}
     \widehat{\eta_1\otimes\eta_2}^{-1}
     \ket{\varphi}_{\mathcal{H}_1\bar{\otimes} \mathcal{H}_2}    
     \ket{\lambda_1}_{\eta_1}\otimes
     \ket{\lambda_2}_{\eta_2}
    d\nu,
    \label{spectralexpansion_adjoint_ket_H_1}\\
     \bra{A^\dagger\varphi}_{\mathcal{H}_1\bar{\otimes}\mathcal{H}_2}
     & = & 
    \displaystyle\int_{Sp(A)}
     \lambda
     \bra{\varphi}_{\mathcal{H}_1\bar{\otimes} \mathcal{H}_2}
     \widehat{\eta_1\otimes\eta_2}^{-1}
         \ket{\lambda_1}_{\eta_1}\otimes
     \ket{\lambda_2}_{\eta_2}
     \bra{\lambda_1}_{\eta_1}\otimes
     \bra{\lambda_2}_{\eta_2}
    d\nu.
    \label{spectralexpansion_adjoint_bra_H_2}
    \end{eqnarray}

It can be easily show using  (\ref{spectralexpansion_ket_H_1}) and (\ref{spectralexpansion_bra_H_1}) that $\big{\{}\ket{\lambda_1}_{\eta_1}\otimes\ket{\lambda_2}_{\eta_2}, \ket{\lambda_1}_{\mathcal{H}_1}\otimes\ket{\lambda_2}_{\mathcal{H}_2}\big{\}}_{\lambda=\lambda_1+\lambda_2\in Sp(A)}$ constitutes
the complete bi-orthogonal base for the quasi-Hermitian composite system, where 
\begin{eqnarray}
    \bra{\lambda_1}_{\mathcal{H}_1}\otimes\bra{\lambda_2}_{\mathcal{H}_2}=\bra{\lambda_1}_{\eta_1}\otimes\bra{\lambda_2}_{\eta_2}\widehat{\eta_1\otimes\eta_2}^{-1}
\text{~and~} 
    \ket{\lambda_1}_{\mathcal{H}_1}\otimes\ket{\lambda_2}_{\mathcal{H}_2}=\widehat{\eta_1\otimes\eta_2}^{-1}\ket{\lambda_1}_{\eta_1}\otimes\ket{\lambda_2}_{\eta_2}.
    \label{eqn:generalized_eigenvector_for_H}
\end{eqnarray}
The completion relation is given by 
\begin{eqnarray}
    I =  \displaystyle\int_{Sp(A)}
     \ket{\lambda_1}_{\eta_1}\otimes
     \ket{\lambda_2}_{\eta_2}
     \bra{\lambda_1}_{\mathcal{H}_1}\otimes
     \bra{\lambda_2}_{\mathcal{H}_2}     
    d\nu
    =
    \displaystyle\int_{Sp(A)}
     \ket{\lambda_1}_{\mathcal{H}_1}\otimes
     \ket{\lambda_2}_{\mathcal{H}_2}
     \bra{\lambda_1}_{\eta_1}\otimes
     \bra{\lambda_2}_{\eta_2}     
    d\nu,
    \label{eqn:completion}
\end{eqnarray}
and
\begin{eqnarray}
\bra{\lambda^{\prime}_1}_{\eta_1}\otimes     \bra{\lambda^{\prime}_2}_{\eta_2}
    \ket{\lambda_1}_{\mathcal{H}_1}\otimes     \ket{\lambda_2}_{\mathcal{H}_2}
    =
    \bra{\lambda^{\prime}_1}_{\mathcal{H}_1}\otimes     \bra{\lambda^{\prime}_2}_{\mathcal{H}_2}
    \ket{\lambda_1}_{\eta_1}\otimes     \ket{\lambda_2}_{\eta_2}
    =
    \Check{\delta}(\lambda_1^\prime-\lambda_1)
    \Check{\delta}(\lambda_2^\prime-\lambda_2)
    \label{eqn:O3-3-6}
\end{eqnarray}
shows the bi-orthogonality where $\Check{\delta}$ is the delta-function specifying
\begin{eqnarray}
    f(\lambda_1^\prime,\lambda_2^\prime)
     = 
    \int_{Sp(A)}
    f(\lambda_1,\lambda_2)\Check{\delta}(\lambda_1^\prime-\lambda_1)
    \Check{\delta}(\lambda_2^\prime-\lambda_2)
    d\nu
    \label{deltafunction}
\end{eqnarray}
for any function $f(\lambda_1,\lambda_2)$.
%

\section{Quasi Hermitian in dual spaces}
\label{sec:4}
As shown in Sec.\ref{sec:3.2}, 
the bra and ket vectors along with their spectral expansions for quasi-Hermitian (non-Hermitian) systems can be formulated within a set of dual and anti-dual spaces.
This result suggests that these dual spaces are available as fundamental spaces when discussing non-Hermitian systems using the bra-ket notation.
Furthermore, the utilization of these dual spaces is a solution to the issue of defining the adjoint of a quasi-Hermitian operator in non-Hermitian composite systems.
In such a system, the physical operator of each isolated system is generally non-Hermitian, $A_i \neq A_i^\dagger$.
Then, $(A^\dagger_1\otimes I_2 + I_1\otimes A^\dagger_2) \neq (A_1\otimes I_2 + I_1\otimes A_2)^\dagger$, and $(A^\dagger_1\otimes I_2 + I_1\otimes A^\dagger_2) \subset (A_1\otimes I_2 + I_1\otimes A_2)^\dagger$ is satisfied\cite{Simon1980}.
Therefore, when the composite operator $A$ is adapted as shown in (\ref{eqn:composite_operator}), we obtain $A^{\dagger} \neq \overline{A^\dagger_1\otimes I_2 + I_1\otimes A^\dagger_2}$, which results in an inconsistency between the definitions of $A$ and $A^\dagger$.
This inconsistency is resolved by extending the operators to the dual spaces.

At first, we identify $(\Phi_1\hat{\otimes} \Phi_2)^{j}$
 with $(\Phi_1{\otimes} \Phi_2)^{j}$ 
under the algebraic isomorphic mapping,  
$(\Phi_1\hat{\otimes} \Phi_2)^{j} \to (\Phi_1{\otimes} \Phi_2)^{j},f\mapsto f|{\Phi_1{\otimes} \Phi_2}$,
for each $j=\prime,\times$\cite{Schaefer1966}.
Hence, we can consider the dual spaces as  
$(\Phi_1{\otimes} \Phi_2)^{\prime}\cup(\Phi_1{\otimes} \Phi_2)^{\times}$ instead of $(\Phi_1\hat{\otimes} \Phi_2)^{\prime}\cup(\Phi_1\hat{\otimes} \Phi_2)^{\times}$.
Note that the tensor product spaces of the dual spaces $\Phi_1^{\prime}\otimes\Phi_2^{\prime}$ and $\Phi_1^{\times}\otimes\Phi_2^{\times}$ are subsets of $(\Phi_1\otimes\Phi_2)^\prime
$ and $(\Phi_1\otimes\Phi_2)^\times$, respectively\cite{Schaefer1966}.
Given a pair of RHSs', $\Phi_i\subset\mathcal{H}_i\subset\Phi_i^{\prime}, \Phi_i^{\times}$ $(i=1,2)$, 
let $F:\mathcal{D}(F)\to \mathcal{H}_1$ and $G:\mathcal{D}(G)\to \mathcal{H}_2$ be  linear operators that are continuous on $(\Phi_1,\tau_{\Phi_1})$ and $(\Phi_2,\tau_{\Phi_2})$ and that satisfy $F\Phi_1\subset\Phi_1$ and $G\Phi_2\subset\Phi_2$, respectively.
As shown in Sec. \ref{sec:2.2}
, these operators have the extensions $\hat{F}:\Phi_1^\prime\cup\Phi_1^\times\to\Phi_1^\prime\cup\Phi_1^\times$ and 
$\hat{G}:\Phi_2^\prime\cup\Phi_2^\times\to \Phi_2^\prime\cup\Phi_2^\times$, given by (\ref{eqn:extension_definition}).
Moreover, their tensor product operator $\hat{F}\otimes \hat{G}$ from $\check{\Phi}$ into $\check{\Phi}$,
where $\check{\Phi}\equiv(\Phi_1^\prime\otimes\Phi_2^\prime)\cup(\Phi_1^\times\otimes\Phi_2^\times)$,
can be extended to the dual spaces $(\Phi_1\otimes\Phi_2)^\prime\cup (\Phi_1\otimes\Phi_2)^\times$.
The extension 
\begin{eqnarray}
\overline{\hat{F}\otimes\hat{G}}:(\Phi_1\otimes\Phi_2)^\prime\cup (\Phi_1\otimes\Phi_2)^\times\to (\Phi_1\otimes\Phi_2)^\prime\cup (\Phi_1\otimes\Phi_2)^\times
\label{eqn:def_extension_dual_0}
\end{eqnarray}
is defined as
\begin{eqnarray}
    \overline{\hat{F}\otimes\hat{G}}(f)(\phi)
    =
    f((F\otimes G) (\phi)),
\label{eqn:def_extension_dual}
\end{eqnarray}
for 
$ f\in (\Phi_1\otimes\Phi_2)^j ~(j=\prime,\times),\phi\in\Phi_1\otimes\Phi_2$.
Clearly, $\overline{\hat{F}\otimes\hat{G}}$ is a linear operator and is the unique extension of $\hat{F}\otimes \hat{G}$ to $(\Phi_1\otimes\Phi_2)^\prime\cup (\Phi_1\otimes\Phi_2)^\times$.
Furthermore, we define the operator  
$\overline{\hat{F}_1\otimes\hat{G}_1\pm \hat{F}_2\otimes\hat{G}_2}$ on $(\Phi_1\otimes\Phi_2)^\prime\cup (\Phi_1\otimes\Phi_2)^\times$ by
$
\overline{\hat{F}_1\otimes\hat{G}_1\pm \hat{F}_2\otimes\hat{G}_2}(f)(\phi)
=
f((F_1\otimes G_1\pm F_2\otimes G_2) (\phi))$,
for 
$ f\in (\Phi_1\otimes\Phi_2)^j ~(j=\prime,\times),\phi\in\Phi_1\otimes\Phi_2$.
We can easily verify the relationship 
$\overline{\hat{F}_1\otimes\hat{G}_1\pm \hat{F}_2\otimes\hat{G}_2}=\overline{\hat{F}_1\otimes\hat{G}_1}\pm \overline{\hat{F}_2\otimes\hat{G}_2}$.
For the $\eta_i$-quasi Hermitian $A_i$, its extension 
$\hat{A}_i:\Phi_i^{\prime}\cup \Phi_i^{\times}\to \Phi_i^{\prime}\cup \Phi_i^{\times}$ can be defined using (\ref{eqn:extension_definition}) $(i=1,2)$.
Using the identity mapping $\hat{I}_i$ on $\Phi_i^{\prime}\cup \Phi_i^{\times}$, which is the extension of each identity $I_i$ on $\Phi_i$, 
the tensor product operators, $\hat{A}_1\otimes\hat{I}_2$ and $\hat{I}_1\otimes\hat{A}_2$, can be defined on $\Phi_1^{j}\otimes\Phi_2^{j}$ for each  $j=\prime,\times$.
Then, the following lemmas hold.
\begin{lemma}
\label{lemma4.1}
The extension $\hat{A}$ of $A$ given by  (\ref{eqn:composite_operator}) can be represented as
\begin{eqnarray}
 \hat{A}=\overline{\hat{A}_1\otimes\hat{I}_2+\hat{I}_1\otimes\hat{A}_2}
    \label{eqn:extension_relation1a}
\end{eqnarray}
on $(\Phi_1\otimes\Phi_2)^\prime\cup (\Phi_1\otimes\Phi_2)^\times$.
\end{lemma}
\begin{proof}
Let us consider the restriction $\hat{A}|\check{\Phi}$, which has the following relation to $\hat{A}_1$ and $\hat{A}_2$\cite{Ohmori2024} :
\begin{eqnarray}
 \hat{A}|{\check{\Phi}}=\hat{A}_1\otimes\hat{I}_2+\hat{I}_1\otimes\hat{A}_2.
    \label{eqn:extension_relation1}
\end{eqnarray}
Using this relation, it can be easily confirmed that for any $f\in(\Phi_1\otimes\Phi_2)^{j}$ and for any $\phi\in\Phi_1\otimes\Phi_2$, $\hat{A}(f)(\phi)=f(A\phi)=\overline{\hat{A}_1\otimes\hat{I}_2+\hat{I}_1\otimes\hat{A}_2}(f)(\phi)$.
\end{proof}
\noindent
Lemma \ref{lemma4.1} shows that the form in (\ref{eqn:composite_operator}), of the linear operator describing the observable of the composite system, is retained in the dual spaces, as shown in (\ref{eqn:extension_relation1a}).
\begin{lemma}
\label{lemma4.2}
Let $\hat{A}^{\dagger}_i:\Phi_i^{\prime}\cup \Phi_i^{\times}\to \Phi_i^{\prime}\cup \Phi_i^{\times}$ be an extension of the adjoint operator $A_i^{\dagger}$ of the $\eta_i$-quasi Hermitian $A_i$, defined by $\hat{A}^\dagger_i(f)(\phi)=f(A_i^\dagger\phi)$ for $f\in\Phi_i^{\prime}\cup\Phi_i^{\times}$ and $\phi\in\Phi_i$ $(i=1,2)$.
Let $\hat{A}^{\dagger}$ be an  extension of the adjoint $A^\dagger$ of $A$ given by (\ref{eqn:composite_operator}) to $(\Phi_1\otimes\Phi_2)^\prime\cup(\Phi_1\otimes\Phi_2)^\times$.
Then we obtain 
\begin{eqnarray}
 \hat{A}^\dagger=\overline{\hat{A}_1^\dagger\otimes\hat{I}_2+\hat{I}_1\otimes\hat{A}_2^\dagger}.
\label{eqn:extension_relation2a}
\end{eqnarray}

\end{lemma}
\begin{proof}
First, we show the relation,
\begin{eqnarray}
 \hat{A}^\dagger|\check{\Phi}=\hat{A}_1^\dagger\otimes\hat{I}_2+\hat{I}_1\otimes\hat{A}_2^\dagger,
    \label{eqn:extension_relation2}
\end{eqnarray}
for the restriction 
$\hat{A}^\dagger|\check{\Phi}$ of $\hat{A}^\dagger$ on 
$\check{\Phi}$.
For each $i=1,2$, the relations $A_i\Phi_i\subset \Phi_i$, $\eta_i\mathcal{D}(A_i)=\mathcal{D}(A^\dagger_i)$, and $\eta_i\Phi_i=\Phi_i$ are satisfied, as $A_i$ is an $\eta_i$-quasi Hermitian.
Using these relations, we obtain $\Phi_i\subset\mathcal{D}(A^\dagger_i)$ for each $i$, which shows that
$\Phi_1\otimes\Phi_2\subset\mathcal{D}(A_1^\dagger)\otimes\mathcal{H}_2=\mathcal{D}(A_1^\dagger\otimes I_2)$.
Similarly, 
$\Phi_1\otimes\Phi_2\subset\mathcal{D}(I_1 \otimes A_2^\dagger)$.
Thus, 
$\Phi_1\otimes\Phi_2\subset\mathcal{D}(A_1^\dagger \otimes I_2)\cap\mathcal{D}(I_1 \otimes A_2^\dagger)=\mathcal{D}(A_1^\dagger\otimes I_2+I_1\otimes A_2^\dagger)$.
As $A^\dagger=\overline{A_1\otimes I_2 +I_1\otimes A_2}^\dagger = (A_1\otimes I_2 +I_1\otimes A_2)^\dagger \supset (A_1^\dagger\otimes I_2+I_1\otimes A_2^\dagger)$,
\begin{eqnarray}
 \hat{A}^\dagger(\phi)=(\hat{A}_1^\dagger\otimes\hat{I}_2+\hat{I}_1\otimes\hat{A}_2^\dagger)(\phi)~~\text{for}~~\forall \phi \in \Phi_1\otimes\Phi_2.
    \label{eqn:extension_relation_in_proof}
\end{eqnarray}
Let $f\in \Phi_1^{\prime}\otimes \Phi_2^{\prime}$; $f$ can be represented by $f=\sum_{j=1}^nf_j^1\otimes f_j^2$ where $f_j^1\in\Phi_1^\prime$ and $f_j^2\in\Phi_2^\prime$ for $j=1,2,\cdots, n~(n<\infty)$. 
Then,  
$(\hat{A}_1^\dagger\otimes\hat{I}_2+\hat{I}_1\otimes\hat{A}_2^\dagger)(f)=\sum_{j=1}^n(\hat{A}_1^\dagger f_j^1 \otimes{f}_j^2+f_j^1\otimes\hat{A}_2^\dagger f_j^2)$.
Therefore, for any $\phi=\sum_{l=1}^m\phi_l^1\otimes\phi_l^2\in\Phi_1\otimes\Phi_2$, on using (\ref{eqn:extension_relation_in_proof}), we have
\begin{eqnarray}
    \hat{A}^\dagger(f)(\phi)   =  
    \sum_{j=1}^n(f_j^1\otimes f_j^2)(A^\dagger \phi)
    & = & \sum_{j=1}^n(f_j^1\otimes f_j^2)
    \Big{[}
    \sum_{l=1}^m(A^\dagger_1\phi_l^1\otimes \phi_l^2+\phi_l^1\otimes A_2^\dagger\phi_l^2)
    \Big{]}
    \nonumber\\
     & = &
     \sum_{j=1}^n
     \sum_{l=1}^m
    \Big{[} f_j^1(A^\dagger_1\phi_l^1)\otimes f_j^2(\phi_l^2)+f_j^1(\phi_l^1)\otimes f_j^2(A_2^\dagger\phi_l^2)
    \Big{]}
    \nonumber\\
     & = &
     \sum_{j=1}^n(\hat{A}_1^\dagger f_j^1 \otimes{f}_j^2+f_j^1\otimes\hat{A}_2^\dagger f_j^2)(\phi)
     \nonumber\\
     & = &
    (\hat{A}_1^\dagger\otimes\hat{I}_2+\hat{I}_1\otimes\hat{A}_2^\dagger)(f)(\phi).
    \label{eqn:extension_relation_calculation}
\end{eqnarray}
Similarly, (\ref{eqn:extension_relation_calculation}) holds for any $f\in \Phi_1^{\times}\otimes \Phi_2^{\times}$ and any $\phi=\sum_{l=1}^m\phi_l^1\otimes\phi_l^2\in\Phi_1\otimes\Phi_2$.
Based on $(\Phi_1^{\prime}\otimes \Phi_2^{\prime})
\cap
(\Phi_1^{\times}\otimes \Phi_2^{\times})=\{0\}
$, the relation (\ref{eqn:extension_relation2}) is obtained.
From (\ref{eqn:extension_relation2}), the relation (\ref{eqn:extension_relation2a}) is derived.
\end{proof}
\noindent
Based on Lemma \ref{lemma4.2}, it is found that in the dual space, the adjoint $\hat{A}^\dagger$ can be characterized as the same form as $\hat{A}$.
%
%
\begin{lemma}
\label{lemma4.4}
For the extensions of the positive operators, $\eta_1$, $\eta_2$, and $\eta_1\otimes\eta_2$, to $(\Phi_1\otimes\Phi_2)^\prime\cup (\Phi_1\otimes\Phi_2)^\times$, 
the following relations are satisfied :
\begin{eqnarray}
    \widehat{\eta_1\otimes\eta_2}
    =
    \overline{\hat{\eta}_1
    \otimes
    \hat{\eta}_2}
    ~~
    \text{and}
    ~~
    \widehat{\eta_1\otimes\eta_2}^{-1}
    =
    \overline{
    \hat{\eta}_1^{-1}
    \otimes
    \hat{\eta}_2^{-1}
    },
\label{eqn:extension_relation_positive_ope1a}
\end{eqnarray}
where $\hat{\eta}_i$ is the extension of $\eta_i$ given by (\ref{eqn:extension_definition}) $(i=1,2)$.
\end{lemma}
\noindent
Note that the relations (\ref{eqn:2-ext3a}) and (\ref{eqn:2-ext3b}) are also obtained from (\ref{eqn:extension_relation_positive_ope1a}) because  $\bra{\varphi_1\otimes\varphi_2}_{\mathcal{H}_1\bar{\otimes}\mathcal{H}_2}$ and $\ket{\varphi_1\otimes\varphi_2}_{\mathcal{H}_1\bar{\otimes}\mathcal{H}_2}$ are in $(\Phi_1\otimes\Phi_2)^\prime$ and $ (\Phi_1\otimes\Phi_2)^\times$, respectively.
\begin{theorem}
\label{theorem4.1}
    Let $A$ be the $\eta_1\otimes\eta_2$-quasi Hermitian operator given by (\ref{eqn:composite_operator}).
    Then, its extension $\hat{A}$ and the adjoint $\hat{A}^\dagger$ of $\hat{A}$ are characterized by the relations (\ref{eqn:extension_relation1a}) and (\ref{eqn:extension_relation2a}), respectively.
Moreover, they satisfy the symmetric relation,
    \begin{eqnarray}
    \bra{\phi}{\hat A^{\dagger}}\ket{\varphi}_{\mathcal{H}_1\bar{\otimes}\mathcal{H}_2}
    =\bra{\phi}\widehat{\eta_1\otimes\eta_2}~
    \hat{A}~\widehat{\eta_1{\otimes}\eta_2}^{-1}\ket{\varphi}_{\mathcal{H}_1\bar{\otimes}\mathcal{H}_2},
    \label{eqn:symmetric_relation_for_extensions}
    \end{eqnarray}
    for any $\phi,\varphi\in\Phi_1\hat{\otimes}\Phi_2$. 
\end{theorem}
\begin{proof}
Generally, for an $\eta$-quasi Hermitian $T$, the symmetry relation $
    \bra{\phi}{\hat T^{\dagger}}\ket{\varphi}_{\mathcal{H}}
    =\bra{\phi}\hat {\eta}
    \hat T\hat {\eta}^{-1}\ket{\varphi}_{\mathcal{H}},
$
between the extensions $\hat{T}$ and $\hat{T}^\dagger$ with respect to $\hat{\eta}$ is satisfied\cite{Ohmori2022}.
This result is applied to the $\eta_1\otimes\eta_2$-quasi Hermitian $A$, and then we obtain (\ref{eqn:symmetric_relation_for_extensions}).
\end{proof}
\noindent
For the expansion $\hat{A}$ and its adjoint $\hat{A}^\dagger$, we obtain
\begin{eqnarray}
    \bra{A\varphi}_{\mathcal{H}_1\bar{\otimes}\mathcal{H}_2}=\bra{\varphi}_{\mathcal{H}_1\bar{\otimes}\mathcal{H}_2}\hat{A},~
    \ket{A\varphi}_{\mathcal{H}_1\bar{\otimes}\mathcal{H}_2}=\hat{A}\ket{\varphi}_{\mathcal{H}_1\bar{\otimes}\mathcal{H}_2},
    \label{eqn:AandAdagger_relations1}
    \\
    \bra{A^\dagger\varphi}_{\mathcal{H}_1\bar{\otimes}\mathcal{H}_2}=\bra{\varphi}_{\mathcal{H}_1\bar{\otimes}\mathcal{H}_2}\hat{A}^\dagger,~
    \ket{A^\dagger\varphi}_{\mathcal{H}_1\bar{\otimes}\mathcal{H}_2}=\hat{A}^\dagger\ket{\varphi}_{\mathcal{H}_1\bar{\otimes}\mathcal{H}_2}.
    \label{eqn:AandAdagger_relations2}
\end{eqnarray}
Therefore, using these relations, 
the spectral expansions (\ref{spectralexpansion_ket_H_1})--(\ref{spectralexpansion_adjoint_bra_H_2}) can be reformulated by $\hat{A}$ and $\hat{A}^\dagger$ on the dual spaces.

\section{Application to non-Hermitian harmonic oscillator}
\label{sec:5}
The representation obtained in the previous sections is applied to a non-Hermitian harmonic oscillator given by a set of conformal $d$-dimensional 
$n$-body systems.
The Hamiltonian of this system is given by\cite{Meljanac2005,Quesne2007}
\begin{eqnarray}
    H=-\frac{1}{2}\sum_{i=1}^{n}\frac{1}{m_i}\nabla^2_i 
    +\frac{1}{2}\omega^2\sum _{i=1}^{n}m_i \vb{x}_i^2
    +\frac{1}{2}c \left( \sum_{i=1}^n\vb{x}_i\cdot \nabla_i +\frac{1}{2}nd \right),
    \label{eqn:5_Hamiltonian}
\end{eqnarray}
where each $m_i$ is the particle mass with the coordinate
$\vb{x}_i=(x_{i1},\cdots,x_{id})$ $(i=1,2,\cdots,n)$ and $c$ is a real constant.
The potential term is eliminated from the Hamiltonian.
Using the following su$(1,1)$ generators $K_0$ and $K_{\pm}$,
\begin{eqnarray}
    K_0 & = & \frac{1}{2\omega}\left( -\frac{1}{2}\sum_{i=1}^n\frac{\nabla_i^2}{m_i}+\frac{\omega^2}{2}\sum_{i=1}^nm_i\vb{x}_i^2\right ),~~
    \label{eqn:5_su(1,1)0}
    \\
    K_{\pm} & = & \frac{1}{2\omega}
    \left[
    \frac{1}{2}\sum_{i=1}^n\frac{\nabla_i^2}{m_i}-\frac{\omega^2}{2}\sum_{i=1}^nm_i\vb{x}_i^2
    \mp \omega
    \left (\sum_{i=1}^n\vb{x}_i\cdot\nabla_i+\frac{1}{2}nd
    \right )
    \right],    
    \label{eqn:5_su(1,1)}
\end{eqnarray}
the model (\ref{eqn:5_Hamiltonian}) produces a simple non-Hermitian 
$\mathcal{PT}$-symmetric oscillator Hamiltonian proposed by Swanson\cite{Swanson2004},
\begin{eqnarray}
    H=\omega \left(a^{\dagger}a+\frac{1}{2}\right ) +\alpha a^2+\beta a^{\dagger,2},
    \label{eqn:Swanson}
\end{eqnarray}
where $K_0=\frac{1}{2}(a^\dagger a+\frac{1}{2})$, $K_{+}=\frac{1}{2}a^{\dagger,2}$, $K_-=\frac{1}{2}a^2$, and we select $c=4\alpha=-4\beta$.
The generators $K_0$ and $K_{\pm}$ given by (\ref{eqn:5_su(1,1)0}) and (\ref{eqn:5_su(1,1)}) represent the Hermitian and non-Hermitian terms, respectively.
Here we assume $\omega^2-4\alpha\beta>0$.
Subsequently, a similarity transformation concerning this Hamiltonian exists.
\begin{eqnarray}
h_\rho & = &  \rho H\rho^{-1}\nonumber\\
        & = & -\frac{1}{2}\mu^2\sum_{i=1}^n\frac{\nabla_i^2}{m_i}+\frac{\nu^2}{2}\sum_{i=1}^nm_i\vb{x}_i^2,   
    \label{eqn:5_similarity_tra}
\end{eqnarray}
where 
%
 $   \mu^2=\frac{1}{1+z}\left[ 1+\frac{z}{\omega |z|}\left(\Omega ^2z^2-\frac{c^2}{4}\right)^{1/2}\right],~
    \nu^2=\frac{1}{1-z}\omega^2\left[1-\frac{z}{\omega |z|}\left(\Omega ^2z^2-\frac{c^2}{4}\right)^{1/2}\right],$
and $\Omega=(\omega^2+\frac{c^2}{4})^{1/2}$.
$z$ is a parameter that determines the Hermiticity of $h_\rho$\cite{Quesne2007}.
Furthermore, the transformation positive parameter $\rho$ in (\ref{eqn:5_similarity_tra}) is given by
\begin{eqnarray}
   \rho & = & \left(
   \frac{\alpha+\beta-\omega z +(\alpha-\beta)\sqrt{1-z^2}}{\alpha+\beta-\omega z -(\alpha-\beta)\sqrt{1-z^2}}\right)^{[2K_0+z(K_++K_-)]/(4\sqrt{1-z^2})}
   \nonumber\\
   & = & \left(
   \frac{\omega z-cq}{\omega z+cq}\right )^{qK_0/(1+z)}
,
    \label{eqn:5_parameter_rho}
\end{eqnarray}
where $q=\frac{\sqrt{1-z^2}}{2}$.
For each $i\in\{1,2,\cdots,n\}$,
the non-Hermitian Hamiltonian, 
$H_i=-\frac{\nabla_i^2}{2m_i}+\frac{\omega^2}{2}m_i\vb{x}_i^2+\frac{c}{2}(\vb{x}_i\cdot\nabla_i+\frac{d}{2})$, comprising (\ref{eqn:5_Hamiltonian}),
has the su$(1,1)$ generator $K_0^{[i]}$ and the corresponding transformation parameter $\rho_i$, which are given by 
\begin{eqnarray}
   K_0^{[i]}=
   \frac{1}{2\omega}\left(
   -\frac{\nabla_i^2}{2m_i}+\frac{\omega^2}{2}m_i\vb{x}_i^2\right )
   \text{~~and~~}
   \rho_i  = 
   \left( \frac{\omega z-cq}{\omega z+cq}\right) ^{qK_0^{[i]}/(1+z)},
    \label{eqn:5_parameters_1}
\end{eqnarray}
respectively.
Thus, we obtain 
\begin{eqnarray}
    K_0=\sum_{i=1}^nK_0^{[i]} \text{~~~and~~~} \rho=\rho_1\rho_2\cdots\rho_n
    \label{eqn:5_parameters_2},
\end{eqnarray}
and 
the similarity transformation shows
\begin{eqnarray}
    h_\rho & = & \rho H\rho^{-1}
     =  (\rho_1\rho_2\cdots\rho_n)H(\rho_1\rho_2\cdots\rho_n)^{-1}\nonumber\\
    & = & 
    \sum_{i=1}^n\rho_iH_i\rho_i^{-1}
     = 
    \sum_{i=1}^nh_{\rho_i},
    \label{eqn:5_similarity_tra_for_rho_i}
\end{eqnarray}
where 
\begin{eqnarray}
\displaystyle h_{\rho_i}=\rho_iH_i\rho_i^{-1}
=\mu^2\left(
    -\frac{\nabla_i^2}{2m_i}\right)
    +\frac{\nu^2}{2}m_i\vb{x}_i^2
    \label{eqn:5_similarity_tra_for_rho_i2}
\end{eqnarray}
is the Hermitian Hamiltonian for the single harmonic oscillator in the $L^2$-space.

We now introduce the RHS treatment and construct a bra-ket formalism using the results obtained in the previous sections.
For simplicity, we set $n=2$ and $d=1$.
Furthermore, we assume that the operators are defined within the bra-ket space (dual spaces).
The RHS exclusively for each harmonic oscillator $h_{\rho_i}$ ($i=1,2$) of the form (\ref{eqn:5_similarity_tra_for_rho_i2}) in the $x$-representation diagram is described by the triplet of Schwartz space $\mathcal{S}=\mathcal{S}(\mathbb{R})$\cite{Bohm1978},
\begin{eqnarray}
    \mathcal{S}\subset L^2 \subset \mathcal{S}^{\prime},\mathcal{S}^{\times}.
    \label{eqn:5_RHS_single_HO}
\end{eqnarray}
The eigenequation for $\hat{h}_{\rho_i}$ is given by 
\begin{eqnarray}
    \hat{h}_{\rho_i} \varphi_{i,l}(x)=\epsilon_l^i\varphi_{i,l}(x)
    \label{eqn:5_eigeneq._for_single_HO}
\end{eqnarray}
with the solutions 
\begin{eqnarray}
    \epsilon_{l}^i &=&  \mu \nu \left( l+\frac{1}{2} \right), 
    \label{eqn:5_solution_of_eigeneq._for_single_HO_0}
    \\
    \varphi_{i,l}(x) &=& \sqrt{\frac{2^{-l}}{l!\mathrm{L_i}\sqrt{\pi} }} e^{-\frac{x^2}{2\mathrm{L_i}^2}} \textrm{H}_l \left(\frac{x}{\mathrm{L_i}} \right),
    \label{eqn:5_solution_of_eigeneq._for_single_HO}
\end{eqnarray}
where $l=0,1,2,\cdots$, $\textrm{H}_l$ is  Hermite polynomials and $\mathrm{L_i}=\sqrt{\mu/m_i\nu}$.
Note that $\varphi_{i,l}(x)$ are the (generalized) eigenfunctions in $\mathcal{S}^{\times}$ constructing the orthonormal base: 
\begin{eqnarray}
    \int_{-\infty}^{\infty}dx~ \varphi_{i,l}(x)\varphi_{i,m}(x)=\delta_{lm},
    ~~
    \sum_{l=0}^\infty \varphi_{i,l}(x')\varphi_{i,l}(x)=\delta(x'-x).
    \label{eqn:onb_for_HO}
\end{eqnarray}
The RHS of the composite system is represented by the tensor product in (\ref{eqn:5_RHS_single_HO}):
\begin{eqnarray}
    \mathcal{S}\hat{\otimes}\mathcal{S}
    \subset 
    L^2 \bar{\otimes} L^2
    \subset 
    (\mathcal{S}\hat{\otimes}\mathcal{S})^{\prime},
    (\mathcal{S}\hat{\otimes}\mathcal{S})^{\times},
    \label{eqn:5_RHS_composite_HO}
\end{eqnarray}
wherein the Hamiltonian corresponding to (\ref{eqn:5_similarity_tra_for_rho_i}) is given by
\begin{eqnarray}
\hat{h}_\rho=\overline{\hat{h}_{\rho_1}\otimes I_2+I_1\otimes \hat{h}_{\rho_2}}.
    \label{eqn:5_2_h_rho}
\end{eqnarray}
It satisfies the eigenequation
\begin{eqnarray}
    \hat{h}_\rho \varphi_{1,l}(x)\otimes \varphi_{2,m}(y)
    =
    (\epsilon_l^1+\epsilon_m^2)\varphi_{1,l}(x)\otimes \varphi_{2,m}(y),
    \label{eqn:5_2_h_rho_eigenequation}
\end{eqnarray}
with the generalized eigenfunctions $\varphi_{1,l}(x)\otimes \varphi_{2,m}(y)\in (\mathcal{S}\hat{\otimes}\mathcal{S})^\times$.
For a wave function $\psi_1(x)\otimes\psi_2(y) \in \mathcal{S}\hat{\otimes}\mathcal{S}$ where $\mathcal{S}\ni \psi_i(x)=\braket{x}{\psi_i}$\cite{Bohm1978}, the spectral expansion is executed in the dual space, as follows:
\begin{align}
    \begin{split}
    \psi_1(x)\otimes\psi_2(y) 
    & = 
    \sum_{l,m}\Bigg\{
    \int_{-\infty}^{\infty}dx^{\prime}\int_{-\infty}^{\infty}dy^{\prime}
    \psi_1(x^{\prime})\otimes \psi_2(y^{\prime})
    \\
    &
    \varphi_{1,l}(x^{\prime})\otimes \varphi_{2,m}(y^{\prime})
    \Bigg\} \varphi_{1,l}(x)\otimes \varphi_{2,m}(y),  
    \end{split}
    \label{eqn:5_2_expansion_h_rho_composite_system1}
\end{align}
\begin{align}
    \begin{split} 
    \hat{h}_{\rho} \psi_1(x)\otimes\psi_2(y) 
    & = 
    \sum_{l,m}(\epsilon_l^1+\epsilon_m^2)\Bigg\{
    \int_{-\infty}^{\infty}dx^{\prime}\int_{-\infty}^{\infty}dy^{\prime}
    \psi_1(x^{\prime})\otimes \psi_2(y^{\prime})
    \\
    & 
     \varphi_{1,l}(x^{\prime})\otimes \varphi_{2,m}(y^{\prime})
    \Bigg\}
    \varphi_{1,l}(x)\otimes \varphi_{2,m}(y).  
    \end{split}
    \label{eqn:5_2_expansion_h_rho_composite_system2}
\end{align}
From the relation of the similarity transformation (\ref{eqn:5_similarity_tra_for_rho_i2}) for each $i$,
 $\hat{H}_i$ becomes an $\hat{\eta}_i$-quasi Hermitian when the positive operator $\hat{\eta}_i$ is given as
\begin{eqnarray}
  \hat{\eta}_i=\hat{\rho}_{i}^2 = 
   \Big{(}\frac{\omega z-cq}{\omega z+cq}\Big{)}^{2qK_0^{[i]}/(1+z)}.
    \label{eqn:5_parameters_eta_i}
\end{eqnarray}
Whereas, by (\ref{eqn:5_similarity_tra}), the Hamiltonian $\hat{H}$ is an $\hat{\eta}$-quasi Hermitian if and only if $\hat{\eta}=\hat{\rho}^2$.
Noting $\hat{\rho}^2=(\hat{\rho}_1\otimes\hat{\rho}_2)^2=\hat{\rho}_1^2 \otimes \hat{\rho}_2^2$, therefore, the Hamiltonian
\begin{eqnarray}
    \hat{H}=\overline{\hat{H}_1\otimes I_2+I_1\otimes \hat{H}_2}
    \label{eqn:5_3_Hamiltonian}
\end{eqnarray}
becomes an $\hat{\eta}_1\otimes\hat{\eta}_2$-quasi Hermitian.
For the adjoint $\hat{H}^\dagger=\overline{\hat{H}_1^\dagger\otimes I_2+I_1\otimes \hat{H}_2^\dagger}$ of $\hat{H}$,
the symmetric relation is satisfied as
\begin{eqnarray}
    \hat{H}^\dagger=\widehat{\eta_1\otimes\eta_2}~
    \hat{H}~\widehat{\eta_1{\otimes}\eta_2}^{-1}.
    \label{eqn:5_3_Hamiltonian_symmetry}
\end{eqnarray}
We consider the $\eta_1\otimes \eta_2$-RHS, induced from the RHS (\ref{eqn:5_RHS_single_HO}),
\begin{eqnarray}
    \mathcal{S}\hat{\otimes}\mathcal{S}
    \subset 
    (L^2 \widetilde{\bar {\otimes}} L^2)_{\eta_1\otimes\eta_2}
    \subset 
    (\mathcal{S}\hat{\otimes}\mathcal{S})^{\prime},
    (\mathcal{S}\hat{\otimes}\mathcal{S})^{\times}.
    \label{eqn:5_eta_RHS_HO}
\end{eqnarray}
The eigenequation of $\hat{H}$ 
is shown by
\begin{eqnarray}
    \hat{H}\Psi_{1,l}(x)\otimes \Psi_{2,m}(y)=(\epsilon_l^1+\epsilon_m^2)\Psi_{1,l}(x)\otimes \Psi_{2,m}(y),
    \label{eqn:5_3_eigenequation_for_H}
\end{eqnarray}
where the generalized eigenfunctions are given by
\begin{eqnarray}
    \Psi_{i,l}(x)=\hat{\rho}_i^{-1}\varphi_{i,l}(x)~~(i=1,2,~l=0,1,2,\cdots).
    \label{eqn:5_3_eigenfunction}
\end{eqnarray}
The completeness and orthogonality relations  constituted by $\{ \Psi_{1,l}(x)\otimes\Psi_{2,m}(y), \hat{\eta}_1(x) \otimes \hat{\eta}_2(y)\Psi_{1,l}(x)\otimes\Psi_{2,m}(y)\}$ are represented as
\begin{eqnarray}
    \sum_{l,m} \Psi_{1,l}(x)\otimes\Psi_{2,m}(y) 
    \hat{\eta}_1(x')\otimes \hat{\eta}_2(y')\Psi_{1,l}(x')\otimes\Psi_{2,m}(y') 
    =  \check{\delta}(x-x')\check{\delta}(y-y'),\label{eq:CompletenessRelation} \\
    \int^\infty_{-\infty} \!dx \int^\infty_{-\infty} \!dy\, \Psi_{1,l}(x)\otimes\Psi_{2,l^\prime}(y) \hat{\eta}_1(x)\otimes\hat{\eta}_2(y)\Psi_{1,m}(x)\otimes\Psi_{2,m^\prime}(y)
     = \delta_{lm}\delta_{l^\prime m^\prime}. 
     \label{eq:OrthogonalityRelation}
\end{eqnarray}
Thus, the spectral expansions for the composite system are obtained as follows:
for the wave function $\psi_1(x)\otimes\psi_2(y)\in\mathcal{S}\hat{\otimes}\mathcal{S}$, 
\begin{align}
    \begin{split} 
    \psi_1(x)\otimes\psi_2(y) 
    & = 
    \sum_{l,m}\Bigg\{
    \int_{-\infty}^{\infty}dx^{\prime}\int_{-\infty}^{\infty}dy^{\prime}
    \psi_1(x^{\prime})\otimes \psi_2(y^{\prime})
    \\
    & \hat{\eta}_1(x')\otimes \hat{\eta}_2(y')
    \Psi_{1,l}(x^{\prime})\otimes \Psi_{2,m}(y^{\prime})
    \Bigg\} \Psi_{1,l}(x)\otimes \Psi_{2,m}(y), 
    \end{split}
    \label{eqn:5_2_expansion_H_composite_system1}
\end{align}
\begin{align}
    \begin{split} 
       \hat{H}\psi_1(x)\otimes\psi_2(y) 
    & = 
    \sum_{l,m}(\epsilon_l^1+\epsilon_m^2)\Bigg\{
    \int_{-\infty}^{\infty}dx^{\prime}\int_{-\infty}^{\infty}dy^{\prime}
    \psi_1(x^{\prime})\otimes \psi_2(y^{\prime})\\
    &\hat{\eta}_1(x')\otimes \hat{\eta}_2(y')
    \Psi_{1,l}(x^{\prime})\otimes \Psi_{2,m}(y^{\prime})
    \Bigg\} \Psi_{1,l}(x)\otimes \Psi_{2,m}(y),  
    \end{split}
    \label{eqn:5_2_expansion_H_composite_system2}
\end{align}
\begin{align}
    \begin{split} 
       \hat{H}^\dagger\psi_1(x)\otimes\psi_2(y) 
    & = 
    \sum_{l,m}(\epsilon_l^1+\epsilon_m^2)\Bigg\{
    \int_{-\infty}^{\infty}dx^{\prime}\int_{-\infty}^{\infty}dy^{\prime}
    \psi_1(x^{\prime})\otimes \psi_2(y^{\prime})\\
    &\Psi_{1,l}(x^{\prime})\otimes \Psi_{2,m}(y^{\prime})
    \Bigg\} \hat{\eta}_1(x)\otimes \hat{\eta}_2(y)\Psi_{1,l}(x)\otimes \Psi_{2,m}(y).  
    \end{split}
    \label{eqn:5_2_expansion_H_composite_system3}
\end{align}
In terms of the bra-ket notation, we set
\begin{eqnarray}
    \ket{\varphi}_{\eta_1\otimes\eta_2}
    =
    \widehat{\eta_1\otimes\eta_2}\ket{\varphi}_{L^2\bar{\otimes} L^2}, 
    \quad 
    \bra{\varphi}_{\eta_1\otimes\eta_2} = \bra{\varphi}_{L^2\bar{\otimes} L^2}\widehat{\eta_1\otimes\eta_2},
\end{eqnarray}
for $\varphi\in \mathcal{S}\hat{\otimes}\mathcal{S}$, and
%
%
we obtain the expansion of the bras, 
\begin{eqnarray}
    \bra{\varphi}_{L^2\bar{\otimes} L^2} 
    &=& \sum_{l,m} \bra{\varphi}_{L^2\bar{\otimes} L^2} \widehat{\eta_1\otimes\eta_2}^{-1} \ket{ \Psi_{1,l}}_{\eta_1}\otimes\ket{ \Psi_{2,m}}_{\eta_2}\bra{ \Psi_{1,l}}_{\eta_1}\otimes\bra{ \Psi_{2,m}}_{\eta_2}, \\
    \bra{\varphi}_{L^2\bar{\otimes} L^2}\hat{H} 
    &=& 
    \sum_{l,m}(\epsilon_l^1+\epsilon_m^2)\bra{\varphi}_{L^2\bar{\otimes} L^2}
   \ket{ \Psi_{1,l}}_{\eta_1}\otimes\ket{ \Psi_{2,m}}_{\eta_2}\bra{ \Psi_{1,l}}_{\eta_1}\otimes\bra{ \Psi_{2,m}}_{\eta_2} \widehat{\eta_1\otimes\eta_2}^{-1}, 
    \\
    \bra{\varphi}_{L^2\bar{\otimes} L^2}\hat{H}^\dagger 
    &=& 
    \sum_{l,m} (\epsilon_l^1+\epsilon_m^2)\bra{\varphi}_{L^2\bar{\otimes} L^2}
    \widehat{\eta_1\otimes\eta_2}^{-1}  \ket{ \Psi_{1,l}}_{\eta_1}\otimes\ket{ \Psi_{2,m}}_{\eta_2}\bra{ \Psi_{1,l}}_{\eta_1}\otimes\bra{ \Psi_{2,m}}_{\eta_2},
\end{eqnarray}
and of the kets,
\begin{eqnarray}
    \ket{\varphi}_{L^2\bar{\otimes} L^2} 
    &=& \sum_{l,m} \bra{ \Psi_{1,l}}_{\eta_1}\otimes\bra{ \Psi_{2,m}}_{\eta_2}\widehat{\eta_1\otimes\eta_2}^{-1} \ket{\varphi}_{L^2\bar{\otimes} L^2} \ket{ \Psi_{1,l}}_{\eta_1}\otimes\ket{ \Psi_{2,m}}_{\eta_2}, \\
    \hat{H}\ket{\varphi}_{L^2\bar{\otimes} L^2} 
    &=& 
    \sum_{l,m}(\epsilon_l^1+\epsilon_m^2)
    \bra{ \Psi_{1,l}}_{\eta_1}\otimes\bra{ \Psi_{2,m}}_{\eta_2}\ket{\varphi}_{L^2\bar{\otimes} L^2} \widehat{\eta_1\otimes\eta_2}^{-1}\ket{ \Psi_{1,l}}_{\eta_1}\otimes\ket{ \Psi_{2,m}}_{\eta_2}, 
    \\
    \hat{H}^\dagger\ket{\varphi}_{L^2\bar{\otimes} L^2} 
    &=& 
    \sum_{l,m} (\epsilon_l^1+\epsilon_m^2)\bra{ \Psi_{1,l}}_{\eta_1}\otimes\bra{ \Psi_{2,m}}_{\eta_2}\widehat{\eta_1\otimes\eta_2}^{-1} \ket{\varphi}_{L^2\bar{\otimes} L^2} \ket{ \Psi_{1,l}}_{\eta_1}\otimes\ket{ \Psi_{2,m}}_{\eta_2}.
\end{eqnarray}
We note that they are consistent with (\ref{spectralexpansion_ket_H_1})--(\ref{spectralexpansion_adjoint_bra_H_2}). 

Although the bra-ket formalism in this study is assumed to be a non-interacting system, interactions can be treated by the perturbation technique\cite{Kato}, for instance, based on the present formalism.

This study has various applications in modern quantum theory, such as open quantum systems, as discussed in this section.
For instance, the tensor product of the RHS can be utilized to mathematically develop the Liouville space formalism and Liouville operator with a non-Hermitian (pseudo-Hermitian) property\cite{Antoiou1998,Gyamfi2020}.
These subjects will be addressed in future studies.

\section{Conclusion}
\label{sec:6}

The mathematical treatment of Dirac’s bra-ket formalism in a quasi-Hermitian composite system is discussed in this study.
The $\eta_1\otimes\eta_2$-rigged Hilbert space is established as the underlying space, based on which the spectral expansions of the bra-ket vectors for the quasi-Hermitian operator can be derived.
In addition to the bra-ket vectors, these expansions are practically executed within a set of dual and anti-dual spaces.
This implies that a mathematical bra-ket description of a quasi-Hermitian composite system can be formulated in these dual spaces.
Furthermore, we show that in dual spaces, the non-Hermitian operator and its adjoint for the composite system are well defined symmetrically, and their symmetric structure is retained.
Our results are applied to a non-Hermitian harmonic oscillator comprising conformal multi-dimensional many-body systems.
%
%




\bigskip

\noindent
{\bf Acknowledgement}

The authors is grateful to
Dr. J. Takahashi for useful comments and encouragement.
Also, the authors is grateful to
Prof. Y.~Yamazaki, Prof. T.~Yamamoto, and Prof. Emeritus A.~Kitada for worthwhile comments and encouragement.
I appreciate Prof. Y.~Yamanaka, 
Prof. H.~Ujino, Prof. I.~Sasaki, 
Prof. H.~Saigo, Prof. F.~Hiroshima, 
Prof. S. Matsutani for comments and encouragement.
This work was supported by the Sasakawa Scientific Research Grant from The Japan Science Society and JSPS KAKENHI Grant Numbers 22K13963 and 22K03442.
We would like to thank Editage (www.editage.jp) for English language editing.\\

\noindent
{\bf Data Availability}

Data sharing is not applicable to this article as no new data were created or analyzed in this study.

\bigskip


\end{document}